\documentclass[final]{elsarticle}

\usepackage{xargs}

\usepackage{xspace}

\usepackage{xstring}

\usepackage{boolexpr}

\usepackage[cmex10]{mathtools}
\usepackage{amssymb}
\usepackage{amsfonts}
\usepackage{mathrsfs}
\usepackage{latexsym}
\usepackage{textcomp}
\usepackage{pifont}

\usepackage{amsthm}

\usepackage{graphics}
\usepackage{lipsum}
\usepackage{tikz}
\usetikzlibrary{arrows}

\usepackage{balance}

\usepackage{wrapfig}

\usepackage[final]{microtype}

\usepackage{listings}
\usepackage{graphics}

\usepackage{amsmath}
\usepackage{amssymb}
\usepackage[mathscr]{eucal}
\usepackage{wasysym}
\usepackage{oxford}
\usepackage{times}
\usepackage{itmath}
\usepackage{booktabs}
\usepackage{mathrsfs}
\usepackage[inline]{enumitem}
\usepackage{ifthen}
\usepackage[boxed]{algorithm}
\usepackage{algorithmic}
\usepackage{comment}
\usepackage{balance}
\usepackage{bm}

\usepackage{color}
\usepackage{graphicx}

\usepackage{xspace}
\usepackage{url}

\usepackage{changebar}

\usepackage{lineno,hyperref}
\modulolinenumbers[10]
\usepackage{fixltx2e}
\usepackage{multicol}
\usepackage{subfigure}

\usepackage[ruled,vlined,linesnumbered,algo2e]{algorithm2e}
\LinesNumbered

\newcommand{\argemp}[2]{\if&#1&\else#2\fi}

\newcommand{\argdef}[2]{\if&#1&#2\else#1\fi}

\newcommand{\argint}[3]{\if&#2&\else#1#2#3\fi}

\newcommand{\argext}[3]{\if&#1&#3\else#1\if&#3&\else#2#3\fi\fi}

\newcommandx{\mthfnt}[3][1=, 2=0]{{
	\IfStrEqCase{#1}
	{%
		{}%
		{#3}%
		{Name}%
		{%
			\IfStrEqCase{#2}
			{%
				{0}{\mathcal{#3}}%
				{1}{\mathscr{#3}}%
				{2}{\mathfrak{#3}}%
				{3}{\mathbb{#3}}%
			}
			[\ensuremath{\clubsuit}]%
		}%
		{Set}%
		{%
			\IfStrEqCase{#2}
			{%
				{0}{\mathrm{#3}}%
				{1}{\mathsf{#3}}%
				{2}{\mathbb{#3}}%
				{3}{\mathbf{#3}}%
			}
			[\ensuremath{\clubsuit}]%
		}%
		{Fun}%
		{%
			\IfStrEqCase{#2}
			{%
				{0}{\mathsf{#3}}%
				{1}{\mathrm{#3}}%
			}
			[\ensuremath{\clubsuit}]%
		}%
		{Rel}%
		{%
			\IfStrEqCase{#2}
			{%
				{0}{\mathit{#3}}%
				{1}{\mathtt{#3}}%
			}
			[\ensuremath{\clubsuit}]%
		}%
		{Sym}%
		{%
			\IfStrEqCase{#2}
			{%
				{0}{\mathtt{#3}}%
				{1}{\mathbf{#3}}%
			}
			[\ensuremath{\clubsuit}]%
		}%
		{Elm}%
		{\mathnormal{#3}}
	}
[\ensuremath{\clubsuit}]%
}}

\newcommand{\mthsub}[1]{\argemp{#1}{\ensuremath{_{\mathnormal{#1}}}}}

\newcommand{\mthsup}[1]{\argemp{#1}{\ensuremath{^{\mathnormal{#1}}}}}

\newcommandx{\mth}[5][1=, 2=0, 4=, 5=]{{\ensuremath{\mthfnt[#1][#2]{#3}\mthsub{#4}\mthsup{#5}}}}

\newcommandx{\mtharg}[6][1=, 2=0, 4=, 5=]{{\mth[#1][#2]{#3}[#4][#5]\ensuremath{\argint{(}{#6}{)}}}}

\newcommand{\mthempty}{\mth[][]}

\newcommand{\mthstyname}{0}
\newcommand{\mthname}[1][]{\mth[Name][\argdef{#1}{\mthstyname}]}

\newcommand{\mthstyset}{0}
\newcommand{\mthset}[1][]{\mth[Set][\argdef{#1}{\mthstyset}]}

\newcommand{\mthstyfun}{0}
\newcommand{\mthfun}[1][]{\mth[Fun][\argdef{#1}{\mthstyfun}]}
\newcommand{\mthargfun}[1][]{\mtharg[Fun][\argdef{#1}{\mthstyfun}]}

\newcommand{\mthstysym}{0}
\newcommand{\mthsym}[1][]{\mth[Sym][\argdef{#1}{\mthstysym}]}

\newcommand{\mthstyelm}{0}
\newcommand{\mthelm}[1][]{\mth[Elm][\argdef{#1}{\mthstyelm}]}

\newcommand{\tuple}[1]
{\ensuremath{\!\argint{\langle}{#1}{\rangle}}}

\newcommand{\txtfnt}[2][]
{{%
		\IfStrEq{#1}{}
		{#2}
		{%
			\StrLeft{#1}{2}[\optbgn]%
			\StrGobbleLeft{#1}{2}[\optend]%
			\IfStrEqCase{\optbgn}
			{%
				{Rm}{\rmfamily\txtfnt[\optend]{#2}}%
				{Sf}{\sffamily\txtfnt[\optend]{#2}}%
				{Tt}{\ttfamily\txtfnt[\optend]{#2}}%
				{Up}{\upshape\txtfnt[\optend]{#2}}%
				{It}{\itshape\txtfnt[\optend]{#2}}%
				{Sl}{\slshape\txtfnt[\optend]{#2}}%
				{Sc}{\scshape\txtfnt[\optend]{#2}}%
				{Md}{\mdseries\txtfnt[\optend]{#2}}%
				{Bf}{\bfseries\txtfnt[\optend]{#2}}%
				{Em}{\emph{\txtfnt[\optend]{#2}}}%
			}
			[\ensuremath{\clubsuit}]%
		}%
	}}
	
	\newcommand{\txtsub}[2][]
	{\argemp{#2}{\ensuremath{_{\text{\txtfnt[#1]{#2}}}}}}
	
	\newcommand{\txtsup}[2][]
	{\argemp{#2}{\ensuremath{^{\text{\txtfnt[#1]{#2}}}}}}
	
	\newcommandx{\txt}[4][1=, 3=, 4=]
	{%
		\ensuremath{\text{%
				\txtfnt[#1]{#2}\ensuremath{\txtsub[#1]{#3}\txtsup[#1]{#4}}%
			}}%
		}
		
		\newcommandx{\txtarg}[5][1=, 3=, 4=]
		{{\txt[#1]{#2}[#3][#4]\argint{(}{#5}{)}}}
		
		\newcommand{\txtstyname}{RmScMd}
		\newcommand{\txtname}[1][]
		{\txt[\argdef{#1}{\txtstyname}]}
		\newcommand{\txtargname}[1][]
		{\txtarg[\argdef{#1}{\txtstyname}]}

			\newcommandx{\AName}[4][1=, 2=, 3=, 4=]{\mthname[#4]{A#3}[#1][#2]}
			\newcommandx{\BName}[4][1=, 2=, 3=, 4=]{\mthname[#4]{B#3}[#1][#2]}
			\newcommandx{\CName}[4][1=, 2=, 3=, 4=]{\mthname[#4]{C#3}[#1][#2]}
			\newcommandx{\DName}[4][1=, 2=, 3=, 4=]{\mthname[#4]{D#3}[#1][#2]}
			\newcommandx{\EName}[4][1=, 2=, 3=, 4=]{\mthname[#4]{E#3}[#1][#2]}
			\newcommandx{\FName}[4][1=, 2=, 3=, 4=]{\mthname[#4]{F#3}[#1][#2]}
			\newcommandx{\GName}[4][1=, 2=, 3=, 4=]{\mthname[#4]{G#3}[#1][#2]}
			\newcommandx{\HName}[4][1=, 2=, 3=, 4=]{\mthname[#4]{H#3}[#1][#2]}
			\newcommandx{\IName}[4][1=, 2=, 3=, 4=]{\mthname[#4]{I#3}[#1][#2]}
			\newcommandx{\JName}[4][1=, 2=, 3=, 4=]{\mthname[#4]{J#3}[#1][#2]}
			\newcommandx{\KName}[4][1=, 2=, 3=, 4=]{\mthname[#4]{K#3}[#1][#2]}
			\newcommandx{\LName}[4][1=, 2=, 3=, 4=]{\mthname[#4]{L#3}[#1][#2]}
			\newcommandx{\MName}[4][1=, 2=, 3=, 4=]{\mthname[#4]{M#3}[#1][#2]}
			\newcommandx{\NName}[4][1=, 2=, 3=, 4=]{\mthname[#4]{N#3}[#1][#2]}
			\newcommandx{\OName}[4][1=, 2=, 3=, 4=]{\mthname[#4]{O#3}[#1][#2]}
			\newcommandx{\PName}[4][1=, 2=, 3=, 4=]{\mthname[#4]{P#3}[#1][#2]}
			\newcommandx{\QName}[4][1=, 2=, 3=, 4=]{\mthname[#4]{Q#3}[#1][#2]}
			\newcommandx{\RName}[4][1=, 2=, 3=, 4=]{\mthname[#4]{R#3}[#1][#2]}
			\newcommandx{\SName}[4][1=, 2=, 3=, 4=]{\mthname[#4]{S#3}[#1][#2]}
			\newcommandx{\TName}[4][1=, 2=, 3=, 4=]{\mthname[#4]{T#3}[#1][#2]}
			\newcommandx{\UName}[4][1=, 2=, 3=, 4=]{\mthname[#4]{U#3}[#1][#2]}
			\newcommandx{\VName}[4][1=, 2=, 3=, 4=]{\mthname[#4]{V#3}[#1][#2]}
			\newcommandx{\WName}[4][1=, 2=, 3=, 4=]{\mthname[#4]{W#3}[#1][#2]}
			\newcommandx{\XName}[4][1=, 2=, 3=, 4=]{\mthname[#4]{X#3}[#1][#2]}
			\newcommandx{\YName}[4][1=, 2=, 3=, 4=]{\mthname[#4]{Y#3}[#1][#2]}
			\newcommandx{\ZName}[4][1=, 2=, 3=, 4=]{\mthname[#4]{Z#3}[#1][#2]}
			
\newcommand{\set}[2]
{\ensuremath{\argint{\{}{\argext{#1}{\allowbreak:\allowbreak}{#2}}{\}}}}

\newcommand{\pow}[1]
{\ensuremath{2^{#1}}}

\newcommand{\card}[1]
{\mthempty{\argint{\vert}{#1}{\vert}}}

\newcommand{\SetN}
{\mthset[2]{N}}

\newcommand{\numcc}[2]
{\{#1, \ldots, #2 \}}

\newcommand{\pthelm}{\pi}
\newcommandx{\pthElm}[3][1=, 2=, 3=]
{\mthelm{\pthelm#3}[#1][#2]}

\newcommand{\apset}{\Phi}
\newcommandx{\APSet}[3][1=, 2=, 3=]
{\mthset{\apset#3}[#1][#2]}

\newcommand{\agnset}{N}
\newcommandx{\AgnSet}[3][1=, 2=, 3=]
{\mthset{\agnset#3}[#1][#2]}

\newcommand{\strset}{Str}
\newcommandx{\StrSet}[3][1=, 2=, 3=]
{\mthset{\strset#3}[#1][#2]}

\newcommand{\strelm}{\sigma}
\newcommandx{\strElm}[4][1=, 2=, 3=, 4=]
{\mthargfun{\strelm#4}[#1][#2]{#3}}
\newcommandx{\LTL}[5][1=, 2=, 3=, 4=, 5=]
{\txtargname{LTL#5{\small\argint{$[$}{#1}{$]$}}}[#2][#3]{#4}\xspace}

\newcommandx{\LTLF}[5][1=, 2=, 3=, 4=, 5=]
{\txtargname{LTL#5{\small\argint{$[$}{#1}{$]$}}}[\textsc{f}][#3]{#4}\xspace}

\newcommandx{\PTL}[5][1=, 2=, 3=, 4=, 5=]
{\txtargname{PTL#5{\small\argint{$[$}{#1}{$]$}}}[#2][#3]{#4}\xspace}

\newcommandx{\LDL}[5][1=, 2=, 3=, 4=, 5=]
{\txtargname{LDL#5{\small\argint{$[$}{#1}{$]$}}}[#2][#3]{#4}\xspace}

\newcommandx{\LDLF}[5][1=, 2=, 3=, 4=, 5=]
{\txtargname{LDL#5{\small\argint{$[$}{#1}{$]$}}}[\textsc{f}][#3]{#4}\xspace}

\newcommandx{\QPLDLF}[5][1=, 2=, 3=, 4=, 5=]
{\txtargname{QP\LDLF#5{\small\argint{$[$}{#1}{$]$}}}[#2][#3]{#4}\xspace}

\newcommandx{\ATLS}[5][1=, 2=, 3=, 4=, 5=]
{\txtargname{Atl$^{\star}$#5{\small\argint{$[$}{#1}{$]$}}}[#2][#3]{#4}\xspace}

\newcommandx{\SL}[5][1=, 2=, 3=, 4=, 5=]
{\txtargname{Sl#5{\small\argint{$[$}{#1}{$]$}}}[#2][#3]{#4}\xspace}

\newcommand{\F}
{\mthsym{F}}

\newcommand{\QPTL}
{\txtname{Q}\PTL}

\newcommand{\Exs}[1]
{\ensuremath{\argint{\langle}{#1}{\rangle}}}

\newcommand{\All}[1]
{\ensuremath{\argint{[}{#1}{]}}}

\newcommandx{\iBG}[5][1=, 2=, 3=, 4=, 5=]
{i\txtargname{BG#5{\small\argint{$[$}{#1}{$]$}}}[#2][#3]{#4}\xspace}

\newcommandx{\iBGF}[5][1=, 2=, 3=, 4=, 5=]
{i\txtargname{BG#5{\small\argint{$[$}{#1}{$]$}}}[\textsc{f}][#3]{#4}\xspace}

\newcommandx{\RMG}[5][1=, 2=, 3=, 4=, 5=]
{\txtargname{RMG#5{\small\argint{$[$}{#1}{$]$}}}[#2][#3]{#4}\xspace}

\newcommandx{\RMGF}[5][1=, 2=, 3=, 4=, 5=]
{\txtargname{RMG#5{\small\argint{$[$}{#1}{$]$}}}[\textsc{f}][#3]{#4}\xspace}

\newcommand{\pto}{\rightharpoonup}

\newcommand{\sink}{\mthfun{sink}}

\newcommand{\strpElm}{\mthelm{\vec{\sigma}}}

\newcommand{\NE}{\mthset{NE}}

\newcommand{\sNE}{\mthset{sNE}}

\newcommand{\winsym}{\mthset{Win}}
\newcommandx{\Win}[3][1=, 2=, 3=]
{\mthset{\winsym#3}[#1][#2]}

\newcommand{\presym}{\mthfun{Pre}}
\newcommandx{\Pre}[3][1=, 2=, 3=]
{\mthset{\presym#3}[#1][#2]}

\newcommand{\eqsym}{\mthfun{Eq}}
\newcommandx{\Eq}[3][1=, 2=, 3=]
{\mthset{\eqsym#3}[#1][#2]}

\newcommandx{\AFW}[5][1=, 2=, 3=, 4=, 5=]
{\txtargname{AFW#5{\small\argint{$[$}{#1}{$]$}}}[#2][#3]{#4}\xspace}

\newcommand{\file}{\mthfun{file}}

% ...

\usepackage[symbol]{footmisc}

\begin{document}

\begin{frontmatter}

\title{Multi-Player Games with LDL Goals over Finite Traces}

\author{Julian Gutierrez$^1$}
\author{Giuseppe Perelli$^{2}$\footnote[2]{This work has been done while being affiliated to the University of Oxford.}}
\author{Michael Wooldridge$^1$}
\address{$^1$Department of Computer Science, University of Oxford \\ $^2$Department of Computer Science, University of G\"{o}teborg}

\begin{abstract}
	{\em Linear Dynamic Logic} on finite traces (\LDLF) is a powerful logic for reasoning about the behaviour of concurrent and multi-agent systems. 
	In this paper, we investigate techniques for both the {\em characterisation} and {\em verification} of equilibria in multi-player games with goals/objectives expressed using logics based on \LDLF. This study builds upon a generalisation of {\em Boolean games}, a logic-based game model of multi-agent systems where players have goals succinctly represented in a logical way. 
	Because \LDLF goals are considered, in the settings we study---Reactive Modules games and iterated Boolean games with goals over finite traces---players' goals can be defined to be regular properties while achieved in a finite, but arbitrarily large, trace. 
	In particular, using {\em alternating automata}, the paper investigates automata-theoretic approaches to the characterisation and verification of (pure strategy Nash) equilibria, shows that the set of Nash equilibria in multi-player games with \LDLF objectives is regular, and provides complexity results for the associated automata constructions.
\end{abstract}

\begin{keyword}
	Games, Temporal Logic, Multi-Agent Systems, Formal Verification. 
\end{keyword}

\end{frontmatter}

%
\section{Introduction}
	\label{sec:int}
	
	{\em Boolean games} (BG~\cite{HHMW01}) are a logic-based model of multi-agent systems where each agent/player $i$ is associated with a goal, represented as a {\em propositional logic} (PL) formula $\gamma_i$, and player $i$'s main purpose is to ensure that $\gamma_i$ is satisfied.
	The strategies and choices for each player $i$ are defined with respect to a set of Boolean variables $\Phi_i$, drawn from an overall set of variables $\Phi$.
	Player $i$ is assumed to have unique control over the variables in $\Phi_i$, in that it can assign truth values to these variables in any way it chooses.
	Strategic concerns arise in Boolean games as the satisfaction of player $i$'s goal $\gamma_i$ can depend on the variables controlled by other players.
	
	{\em Reactive Modules games} (RMG~\cite{GHW17}) and {\em iterated Boolean games} (iBG~\cite{GHW15b}) generalise Boolean games by making players interact with each other for infinitely many rounds.
	As in the standard (one-shot or one-round) setting described above, there are $n$ players each of whom uniquely controls a subset of Boolean variables and defines the achievement of a particular goal formula $\gamma_i$ satisfied.
	However, in RMGs and iBGs, players' goals $\gamma_i$ are {\em Linear Temporal Logic} formulae (\LTL~\cite{Pnu77}), rather than PL formulae, which are naturally interpreted over infinite sequences of valuations of the variables in $\Phi$; thus, in both RMGs and iBGs, such infinite sequences of valuations represent the plays of these games.%
	\footnote{Although similar, iBGs and RMGs have a number of differences that will be discussed later in the paper.}
	
	Even though RMGs, iBGs, and conventional Boolean games are logic-based models of multi-agent systems, they capture players' goals---and therefore the desired behaviour of the underlying multi-agent systems---in radically different ways: whereas Boolean games have PL goals (which are naturally evaluated over one-round games), RMGs and iBGs have \LTL goals (which are naturally evaluated over games with infinitely many rounds), encompassing two extremes of the landscape when considering repeated games.
	However, there are games, systems, or situations where goals evaluated after an unbounded, but certainly finite, number of rounds should, or must, be considered~\cite{MS03,SRM18}.
	
	In this paper we fill this gap and define and investigate {\em multi-player games with goals over finite traces}, which are games where players' goals can be satisfied/achieved after a finite, but arbitrarily large, number of rounds.
	More specifically, the goals in these games are given by {\em Linear Dynamic Logic} formulae (\LDLF) which are evaluated over finite sequences of valuations of the variables in $\Phi$, that is, over {\em finite traces} of valuations, instead of PL formulae (as in BGs) or \LTL formulae (as in RMGs and iBGs). 
	Thus, while in game with goals over finite traces a play still is an infinite trace of valuations, the satisfaction of a player's goal may occur after an unbounded but finite number of rounds.
	This sharply contrasts with the case of goals given by \LTL formulae ({\em e.g.}, as in iBGs and RMGs), where it may be that a player's objective is satisfied only after considering the full infinite trace of valuations. 
	This simple feature has significant implications, since rather complex automata constructions for the analysis of logics and games over infinite traces may become conceptually simpler under this new semantic (logic-based) framework.
	More importantly, this key observation allows one to define an automata model
	that exactly characterises the set of Nash equilibria in games with goals
	given by regular objectives.
	\\ \indent
	There are several reasons to consider \LDLF goals.
	\LDLF offers great expressive power to our logic-based framework, which is indeed equivalent to monadic second-order logic (MSO).
	On the other hand, \LTL interpreted on finite traces (\LTLF) is as expressive as first-order logic (FOL) over finite traces~\cite{DV13}.
	This, in turn, implies that, over finite traces, while with \LTLF we can only
	describe star-free regular languages/properties, with \LDLF we can describe
	all regular languages/properties---that is, the properties and languages that
	can be described by regular expressions or finite state automata.
	Nevertheless, the automata-theoretic approach and complexity results for solving their related decision problems are equivalent, showing that the gain in expressiveness is achieved for free.
	%
	In this paper, we first define multi-player games with \LDLF goals and then investigate their main game-theoretic properties using a new automata-theoretic approach to reasoning about Nash equilibria.
	Our technique to reason about equilibria builds on automata constructions originally defined to reason about \LDLF formulae~\cite{DV13,DV15}.
	Using this automata-theoretic technique we show a number of subsequent verification and characterisation results, as follows.
	\\ \indent
	Firstly, we show that \emph{checking} whether some strategy profile is a Nash
	equilibrium of a game is a PSPACE-complete problem, thus no harder than
	\LDLF satisfiability~\cite{DV13}.
	Secondly, we focus on the \textsc{NE-Nonemptiness} problem---which  
	asks for the existence of a Nash equilibrium in a multi-player game 
	succinctly specified by a set of Boolean variables and \LDLF formulae---and 
	show that deciding if a multi-player game with \LDLF goals (whether RMG or iBG) has a Nash equilibrium
	can be solved in 2EXPTIME, thus no harder than solving \LDLF
	synthesis~\cite{DV15}.
	The automata technique we use for this problem also shows that the set of Nash
	equilibria in these games is $\omega$-regular and can therefore be {\em
		characterised} using alternating automata.
	Thirdly, we also provide complexity results for the main decision problems
	related to the equilibrium analysis of these games with respect to extensions and restrictions of the initially studied framework.
	In particular, we show that a small extension of the goal language, which
	we call \emph{Quantified-Prefix Linear Dynamic Logic} (\QPLDLF),
	has the same automata-theoretic characteristics as \LDLF, and so it can be
	studied using the same techniques.
	Moreover, \LDLF synthesis can be expressed in \QPLDLF, 
	ensuring 2EXPTIME-completeness. 
	\\ \indent
	Regarding restrictions on the general framework, we first focus on the problem
	of reasoning with memoryless strategies.
	We show, using an automata construction, that the set of Nash equilibria for
	this games is also $\omega$-regular.
	However, an alternative procedure for this problem, not based on automata, 
	shows that improved complexity can be obtained when compared with the 
	standard automata techniques to reason about \LDLF. 
	Another restriction on strategies considered in the paper is the one of myopic
	strategies (which can be used to define all beneficial deviations in a game), 
	in which players perform actions that are independent of the current state of
	the game execution.
	We show that games with such a restriction can be solved in EXPSPACE. 
	We also consider the much more stable solution concept of
	\emph{strong Nash equilibrium}, 
	where sets of players in the game are allowed to jointly deviate, 
	and provide an adaptation of the
	automata-based approach that retains the language characterisation
	and complexity properties of Nash equilibrium.
	\\ \indent
	A key contribution of this work is that our automata-theoretic approach
	features two novel properties, within the same reasoning framework.
	Firstly, it shows that {\em checking} the existence of Nash equilibria can be 
	reduced to a number of \LDLF synthesis and satisfiability
	problems---generalising ideas initially used to reason about \LTL
	objectives~\cite{GHW15a}.
	Secondly, our automata constructions provide reductions where not only
	non-emptiness but also language equivalence is preserved.
	This additionally shows that the set of Nash equilibria in infinite games with 
	regular goals is an $\omega$-regular set, to the best of our knowledge, 
	a semantic {\em characterisation} not previously known, and which do not
	immediately follows from other representations of Nash equilibria---see, {\em
		e.g.}, \cite{ChatterjeeHP10,FismanKL10,KPV14,MMPV14,KPV15}.
	
	\subsection*{Motivation and Previous Work}
	While studying either multi-player games or \LDLF is interesting in itself, 
	from an AI perspective, our main motivation comes from applications to 
	multi-agent systems. In particular, it has been shown that in many scenarios, 
	for instance in the context of 
	planning AI systems~\cite{DV13,DV15}, while logics like LTL, 
	or even LTL over finite traces (LTL$_F$), 
	can be used to reason about the behaviour of agents in such AI systems,  these 
	logics are not powerful enough to express in a satisfactory way the main features of 
	agents in such a context.
	In order to illustrate the use of \LDLF, and motivate even further our work, 
	we will present an example in the next section, where some of the goals 
	either are not expressible in LTL%
	\footnote{This fact can be proved from the result that \LDLF is equivalent to Monadic-Second Order logic, while LTL is equivalent to First-Order Logic~\cite{DV13}} or have a more intuitive specification in \LDLF than in LTL.
	Together with applications to planning AI systems (see~\cite{DV13,DV15}), 
	this is an example of another 
	instance where one can see an advantage of using a game with \LDLF goals 
	over a game with LTL goals, instead. 
	
	Moreover, regarding previous work, while our model builds on RMGs~\cite{GHW17} and iBGs~\cite{GHW15b}, where goals are given by LTL formulae, there are at least two main differences with such work. Firstly, we study scenarios that consider memoryless and myopic strategies, for which results on iBGs have not been investigated%
	\footnote{Memoryless and myopic strategies were already studied for RMGs in~\cite{GPW16}.}.
	Secondly, and most importantly, the tools developed in this paper to obtain most of our complexity and characterisation results, are {\em technically} remarkably different from those used for RMGs and iBGs, specifically, with respect to the techniques used in~\cite{GutierrezHW13,GHW15a,GHW17}. 
	To be more precise, for RMGs and iBGs the main question is reduced to rational synthesis~\cite{FismanKL10}, whose solution goes via a parity automaton characterising formulae of an extension of Chatterjee {\em et al}'s Strategy Logic~\cite{ChatterjeeHP10}, which leads to an automata construction that can be further optimised if computing Nash equilibria is the only concern.
	Instead, in our case, we reduce the problem {\em directly} to a question of  automata constructed in a different way.
	As a consequence, we provide a new set of automata constructions which do not rely on nor relate to those used in rational synthesis, {\em i.e.}, those used to solve RMGs and iBGs.
	Our automata constructions are also different from those used by De Giacomo and Vardi in~\cite{DV13,DV15,DV16}, as described next. 
	
	In~\cite{DV13,DV15,DV16}, De Giacomo and Vardi study the satisfiability and synthesis problems for \LDLF, with and without imperfect information.
	Because of the (game-theoretic) nature of these two problems, their automata constructions deal with two-player zero-sum turn-based scenarios only. 
	Instead, in our case, we deal with multi-player general-sum concurrent scenarios. 
	This difference leads to a completely different technical treatment/manipulation of the automata that can be 
	initially constructed from \LDLF formulae. In fact, their automata constructions and ours are the same 
	only up to the point where \LDLF formulae are translated into automata---that is, the very first step 
	in a long chain of constructions.  
	Moreover, since De Giacomo and Vardi study synthesis and satisfiability problems (represented by two-player games), 
	whereas we study Nash equilibria (in the context of multi-player games), 
	we are required to have a different technical treatment of the automata involved 
	in the solution of the problems investigated in this paper.
	
\section{Formal Framework}
	\label{sec:iBGf}
	
	\subsection{Linear Dynamic Logic on Finite Traces}
		
		In this paper, we consider Linear Dynamic Logic on Finite Traces (\LDLF), a temporal logic introduced in~\cite{DV13} in order to reason about systems whose behaviour can be characterised by sets of finite traces, that is, finite sequences of valuation for the variables of the system.
		
		\begin{definition}[Syntax]
			The syntax of \LDLF is as follows: 
			
			\begin{center}
				
				$\varphi := p \mid \neg \varphi \mid \varphi \wedge \varphi \mid
				\varphi \vee \varphi \mid \Exs{\rho} \varphi \mid \All{\rho}
				\varphi$
				
				$\rho := \psi \mid \varphi? \mid \rho + \rho \mid \rho ;
				\rho \mid \rho^*$,
			\end{center}
			where $p$ is an atomic proposition in $\APSet$; $\psi$ denotes a propositional formula over the atomic propositions in $\APSet$.
		\end{definition}
		
		The symbol $\rho$ denotes path expressions, which are regular expressions over propositional formulae $\psi$, with the addition of the test construct $\varphi?$ from propositional dynamic logic (PDL); and $\varphi$ stands for \LDLF formulae built by applying Boolean connectives and the modal connectives. 
		Tests are used to insert checks for satisfaction of additional \LDLF formulae.
		
		The classic \LTLF operators can be defined as follows: $\ltlnext \varphi \equiv \Exs{\top} \varphi$; $\sometime \varphi \equiv \Exs{\top^*} \varphi$; $\always \varphi \equiv \All{\top^*} \varphi$; $\varphi_{1} \until \varphi_{2} \equiv \Exs{(\varphi_{1}?)^*} \varphi_{2}$.
		
		\LDLF formulae are interpreted over finite traces of the form $\pthElm: \numcc{0}{t} \to \pow{\APSet}$ and an integer $i \in \numcc{0}{t}$.
		
		\begin{definition}[Semantics]
			The semantics of \LDLF formulae is as follows:
			\begin{itemize}
				
			\item
				$\pthElm, i \models \neg \varphi$ if $\pthElm, i \not\models \varphi$;
				
			\item
				$\pthElm, i \models \varphi_{1} \wedge \varphi_{2}$ if $\pthElm, i
				\models \varphi_{1}$ and $\pthElm, i \models \varphi_{2}$;
				
			\item
				$\pthElm, i \models \varphi_{1} \vee \varphi_{2}$ if $\pthElm, i
				\models \varphi_{1}$ or $\pthElm, i \models \varphi_{2}$;

			\item
			$\pthElm, i \models \Exs{\alpha} \varphi$ if there exists $j \in
			\numcc{i}{t}$ such that $(i, j) \in \RName(\alpha, \pthElm)$ and
			$\pthElm, j \models \varphi$;
			
			\item
			$\pthElm, i \models \All{\alpha} \varphi$ if for all $j \in
			\numcc{i}{t}$, if $(i, j) \in \RName(\alpha, \pthElm)$ then
			$\pthElm, j \models \varphi$;
			
		\end{itemize}
		
		where $\RName(\alpha, \pthElm) \subseteq \SetN \times \SetN$ is
		recursively defined by 
		
		\begin{itemize}
			\item
			$\RName(\psi, \pthElm) = \set{(i, i + 1)}{\pthElm(i) \models \psi}$;
			
			\item
			$\RName(\varphi?, \pthElm) = \set{(i, i)}{\pthElm, i \models
				\varphi}$;
			
			\item
			$\RName(\alpha_{1} + \alpha_{2}, \pthElm) = \RName(\alpha_{1},
			\pthElm) \cup \RName(\alpha_{2}, \pthElm)$;
			
			\item
			$\RName(\alpha_{1} ; \alpha_{2}, \pthElm) = \set{(i, j)}{\exists k
				\in \numcc{i}{j}. (i, k) \in \RName(\alpha_{1}, \pthElm) \wedge (k,
				j) \in \RName(\alpha_{2}, \pthElm)}$;
			
			\item
			$\RName(\alpha^*, \pthElm) = \{(i, i)\} \cup \set{(i, j)}{
				\exists k \in \numcc{i}{j}. (i, k) \in \RName(\alpha, \pthElm)
				\wedge (k, j) \in \RName(\alpha^*, \pthElm)}$.
		\end{itemize}
	
	By $\pi \models \varphi$, we denote the fact that $\pi, 0 \models \varphi$.

	\end{definition}

	\subsection{Iterated Boolean Games}
	
	We now introduce \emph{iterated Boolean games with goals over finite traces} (\iBGF), 
	which build upon the framework of iBGs~\cite{GHW15b}. 
	In an \iBGF, players' goals are given by \LDLF formulae interpreted on infinite paths of valuations over a given set of Boolean variables.

	An \iBGF is a tuple $\GName = \tuple{\AgnSet, \APSet, \APSet_{1}, \ldots, \APSet_{n}, \gamma_{1}, \ldots, \gamma_{n}}$, where $\AgnSet = \numcc{1}{n}$ is a set of players, $\APSet$ is a set of Boolean variables, partitioned into $n$ sets $\APSet_{1}, \ldots, \APSet_{n}$, and the goals $\gamma_1,\ldots,\gamma_n$ of the game are \LDLF formulae over $\APSet$.
	In an \iBGF each player $i$  is assumed to control a set of
	propositional variables $\APSet_{i}$, in the sense that player $i$ has the
	power to set the values (true ``$\top$'' or false ``$\bot$'') of each of the
	variables in $\APSet_{i}$.
	An \emph{action} for player $i$ is a possible valuation $v_{i} \in
	\pow{\APSet_{i}}$.
	An \emph{action vector} $\vec{v} = \tuple{v_{1}, \ldots v_{n}}$
	is a collection of actions, one for each player in the game. 
	Every action vector determines an overall valuation for the variables in
	$\APSet = \bigcup_{i = 1}^{n} \APSet_{i}$ of the game.
	An \iBGF is played for an infinite number of rounds, as an iBG, but the goals of the game, which are \LDLF formulae, are interpreted on finite traces of such an infinite run. 
	As a consequence, the satisfaction of a player's goal in the game must occur after a finite, yet arbitrarily large, number of rounds. 
	Due to this, we need to define how an \LDLF formula is satisfied on an \iBGF.
	The most natural way to do so, also implicitly followed in~\cite{DV15}, is
	to say that an infinite play $\pi^{\infty}$%
	\footnote{In this paper, we denote the infinite plays by $\pi^{\infty}$ in order to distinguish them from the finite plays, simply denoted by $\pi$.} satisfies an \LDLF formula
	$\varphi$ if and only if there exists $k \in \SetN$ such that the prefix up to $k$ of $\pi^{\infty}$, denoted by $\pi^{\infty}_{< k}$, satisfies $\varphi$, i.e., $\pi^{\infty}_{< k} \models \varphi$.
	We also write $\pi^{\infty} \models \varphi$ if there is $k \in \SetN$ such
	that $\pi^{\infty}_{< k} \models \varphi$.
	
	Observe that this definition allows a formula and its negation to be
	satisfied on the same infinite play.
	Consider, for example, the \LDLF formula $\varphi = \Exs{\top^{*}} p$, 
	which is satisfied by all and only finite traces ending with a state labelled
	by $p$, and the infinite play $\pi^{\infty} = (\bar p p)^{\omega}$
	which toggles the value of $p$ infinitely often.
	Clearly, $\varphi$ is satisfied on every prefix of $\pi^{\infty}$ of even length,
	while $\neg \varphi$ is satisfied on every prefix of $\pi^{\infty}$ of odd
	length.
	Thus, we obtain that $\pi^{\infty} \models \varphi$ and $\pi^{\infty} \models \neg
	\varphi$.
	This means that the notion of satisfaction given by $\pi^{\infty} \models \neg \varphi$ cannot be used in place of $\pi^{\infty} \not\models \varphi$, as they are not equivalent.
	We discuss this later in the paper, and show that a small extension of
	\LDLF, which allows one to quantify over finite plays in a game, can be used
	to write formulae $\psi$ that equals to the non-satisfaction of $\varphi$,
	i.e., where $\pi^{\infty} \models \psi$ if and only if $\pi^{\infty} \not\models \varphi$.
	
	Observe that the set of models of an \LDLF formula $\varphi$ is of the form $\alpha \cdot (\pow{\APSet})^{\omega}$, with $\alpha$ representing the set of finite traces satisfying $\varphi$.
	In~\cite{DV13} it has been proven that $\alpha$ can be represented by a regular expression.
	This implies, as shown later, that the set of infinite plays satisfying an \LDLF formula can be described in terms of an nondeterministic B{\"u}chi word automaton (NBW) built upon the nondeterministic finite word automaton (NFW) accepting $\alpha$, in which accepting states are constructed so that they are sinks with a self-loop. 
	This makes the expressive power of \LDLF be incomparable with that of \LTL
	when considering infinite words---i.e., infinite plays.
	Indeed, on the one hand, it is known that \LTL cannot express the
	$\omega$-regular expression $(p \cdot \pow{\APSet})^{*} \cdot
	(\pow{\APSet})^{\omega}$~\cite{KS05}, and, on the other hand, \LDLF cannot
	express the \LTL formula $\always \sometime p$ (``always eventually'' $p$), 
	for which every NBW accepting the
	set of models cannot be of the form described above.
	Now, in order to illustrate the concepts introduced so far, 
	and further motivate the \iBGF framework, 
	we present an example where the need for \LDLF goals plays an essential role. 
	\begin{example}
		\label{exm:filesh}
		Consider a file-sharing network composed by a protocol manager and 2 clients who want to share $\file_{1}$ of size $n_1$ packets and $\file_{2}$ of size $n_2$ packets, respectively.
		The clients want to eventually download the other client's file, while the protocol manager wants this transfer of information to happen in a fair way.
		For instance, the manager wants client~1 to always upload in odd time-steps of the communication, while client~2 to always upload in the even time-steps of the communication protocol.
		Moreover, the download of a given file can be marked as completed only after the whole number of its packets has been uploaded by the other party. 
		
		We can represent this protocol by means of a three agent game with $\AgnSet = \{0, 1, 2\}$ in which $\Phi_{0} = \{d_{1}, d_{2}\}$, $\Phi_{1} = \{u_{1} \}$, and $\Phi_{2} = \{u_{2} \}$.
		Variable $u_{i}$ being true means that a single packet of $\file_{i}$ has been uploaded by agent $i$, while variable $d_{i}$ being true means that the download of $\file_{i}$ has been completed.
		Regarding the goals of the agents, we have the following (\LDLF) formulae.
		The two clients $1$ and $2$ want to eventually download $\file_{2}$ and $\file_{1}$, respectively.
		Thus, we have $\gamma_{1} = \Exs{\top^{*}} d_{2}$ and $\gamma_{2} = \Exs{\top^{*}} d_{1}$.
		Regarding the goal $\gamma_{0}$ of the protocol manager, this has to include several requirements.
		First of all, it requires that client~1 always uploads in odd time-steps until the download of $\file_{1}$ has been completed, while client~2 does the same on even time-steps.
		We can represent these properties with the following \LDLF formulae: 
		$\gamma_{upl_1} = \Exs{(u_{1};\top)^{*}} d_{1}$ and $\gamma_{upl_2} = \Exs{(\top;u_{2})^{*}} d_{2}$.
		Note that client~1 has no requirement on the even time-steps, as neither client~2 on odd time-steps.
		The reader might notice that the properties $\gamma_{upl_1}$ and $\gamma_{upl_2}$ can be represented neither in \LTL nor in \LTLF, which is the finite trace version of \LTL~\cite{KS05}.
		In addition to this, the protocol manager is in charge of marking the files as completely downloaded at the right time of the execution.
		This means that variable $d_{i}$ has to be set to true once the whole amount of packets of $\file_{i}$ has been uploaded and not before.
		We can specify this requirement with the following \LDLF formulae: $\gamma_{com_1} = \All{((\neg u_{1})^{*} ; u_{1})^{n_1}} d_{1}$ and $\gamma_{com_2} = \All{((\neg u_{2})^{*} ; u_{2})^{n_2}} d_{2}$.
		In order to avoid that the protocol manager wrongly marks the download of a file as completed, 
		we also have the following requirements: $\gamma_{incom_1} = \bigwedge_{n < n_1} \neg \Exs{((\neg u_{1})^{*}; u_{1};(\neg u_{1})^{*})^{n}} d_{1}$ and $\gamma_{incom_2} = \bigwedge_{n < n_2} \neg \Exs{((\neg u_{2})^{*}; u_{2};(\neg u_{2})^{*})^{n}} d_{2}$.
		The goal of the protocol manager is therefore given by the conjunction of all these conditions: $\gamma_{0} = \gamma_{upl_1} \wedge \gamma_{upl_2} \wedge \gamma_{com_1} \wedge \gamma_{com_2} \wedge \gamma_{incom_1} \wedge \gamma_{incom_2}$.
		To see that the system we have just described/designed has a stable behaviour, 
		from a game-theoretic point of view, we need the concepts of strategies and 
		Nash equilibria, which are introduced next. 
		Then, we will review this example again later on. 
	\end{example}
	
	\subsection{Simple Reactive Modules Language Games}
	\emph{Simple Reactive Modules}~\cite{HLW06} is a model specification language that is based on Reactive Modules~\cite{AH99b} and has been used to describe multi-player games with \LTL goals~\cite{GPW16,GHW17}. 
	Reactive Modules games (RMG) are an extension of iBGs in which one can specify constraints on the power that a player has over the variables that such a player controls%
	\footnote{Iterated Boolean Games result in the special case of a RMG in which no constraint is specified for every agent.}.
	In addition, one can specify multi-player games directly in a high-level description language (which one can then use as the input of a verification tool -- Reactive Modules are used, {\em e.g.}, in MOCHA~\cite{AHMQRT98} and PRISM~\cite{kwiatkowska:2009a}), which is more convenient from a user point of view for modelling purposes. 
	
	In an RMG, an agent is mapped to a reactive module, a machinery that dynamically specifies the choices that are available to the associated agent.
	Formally, a reactive module consists of:
	
	\begin{itemize}
		\item[{\em (i)}]
		an \emph{interface}, which defines the module's name and the set of
		Boolean variables under the \emph{control} of the module; and
		
		\item[{\em (ii)}]
		a number of \emph{guarded commands}, which define the choices available 
		to the module at every state.
	\end{itemize}

	Guarded commands are of two kinds: those used for \emph{initialising} the variables under the module's control (\textbf{init} guarded commands), and those for \emph{updating} these variables subsequently 
	(\textbf{update} guarded commands).
	A guarded command has two parts: a condition part (the ``guard'') and 
	an action part, which defines how to update the value of (some of) the 
	variables under the control of a module.
	The intuitive reading of a guarded command $\phi \ASSIGN \mthfun{a}$ is ``if the
	condition $\phi$ is satisfied, then \emph{one of the choices available to
	the module is to execute the action $\mthfun{a}$}''. 
	We note that the truth of the guard $\phi$ does not mean that $\mthfun{a}$
	\emph{will} be executed: only that such a command is \emph{enabled} for 
	execution---it 
	\emph{may be chosen}. 

	Formally, a guarded command $g$ over the set of Boolean variables $\Phi$ is an expression
	\[
	\phi \ASSIGN x_1 := \psi_1; \cdots; x_k := \psi_k
	\]
	where $\phi$ (the guard) is a propositional logic formula over $\Phi$, each $x_i$ is a controlled variable, and each $\psi_i$ is a propositional logic formula over $\Phi$.
	Let $guard(g)$ denote the guard of~$g$.
	Thus, in the above rule, $guard(g) = \phi$.
	We require that no variable appears on the left hand side of two assignment
	statements in the same guarded command.  
	We say that $x_1, \ldots, x_k$ are the \emph{controlled variables} of $g$,
	and denote this set by $ctr(g)$. 
	If no guarded command of a module is enabled, the values of all variables 
	in~$ctr(g)$ are left unchanged; in \acro{srml} notation, if needed, 
	$\SKIP$ will refer to this particular case.
	
	Formally, an \acro{srml} module, $m_i$, is 
	defined as a triple:
	\[
	m_i = \tuple{\Phi_{i}, I_{i}, U_{i}},   
	\]
	where:
	\begin{itemize}
	\item 
	$\Phi_{i} \subseteq \Phi$ is the (finite) set of variables controlled by
	$m_i$;
	\item 
	$I_i$ is a (finite) set of \emph{initialisation} guarded commands, such that
	for all $g \in I_i$, we have $ctr(g) \subseteq \Phi_i$; and 
	\item 
	$U_i$ is a (finite) set of \emph{update} guarded commands, such that for all
	$g \in U_i$, we have $ctr(g) \subseteq \Phi_i$.
	\end{itemize}
	
	Modules can be composed in an intersection manner as follows.
	For two modules $m_{1} = (\Phi_{1}, I_{1}, U_{1})$ and $m_{2} = (\Phi_{2}, I_{2}, U_{2})$ with $\Phi_{1} \cap \Phi_{2} = \emptyset$, the product module is $m_{1} \otimes m_{2} = (\Phi_{1} \cup \Phi_{2}, I_{1} \otimes I_{2}, U_{1} \otimes U_{2})$ where the $\otimes$-operator over sets of guards $G_{1}$ and $G_{2}$ is defined as the set of guards $g$ such that there exist $g_{1} \in G_{1}$ and $g_{2} \in G_{2}$ of the form $\phi_{1} \ASSIGN x_1^{1} := \psi_1^{1}; \cdots; x_k^{1} := \psi_{k_{1}}^{1}$ and $\phi_{2} \ASSIGN x_1^{2} := \psi_1^{2}; \cdots; x_k^{2} := \psi_{k_{1}}^{2}$, respectively, such that $g$ is of the form:
	
	\[
	\phi_{1} \wedge \phi_{2} \ASSIGN x_1^{1} := \psi_1^{1}; \cdots; x_k^{1} := \psi_{k_{1}}^{1}; x_1^{2} := \psi_1^{2}; \cdots; x_k^{2} := \psi_{k_{1}}^{2}
	\]
	
	An \acro{srml} \emph{game} over \LDLF goals (\RMGF) is then defined to be a tuple: 
	
	\[\GName= \tuple{N, \Phi, M, \eta, \gamma_{1}, \ldots, \gamma_{n}} \]
	where $N = \{1, \ldots, n \}$ is a set of agents, $\Phi$ is a set of Boolean variables, $M = \{m_1, \ldots m_M\}$ is a set of modules such that $\Phi_{m_{1}}, \ldots \Phi_{m_{M}}$ forms a partition of $\Phi$ (so every variable in $\Phi$ is controlled by some module, and no variable
	is controlled by more than one module), and $\eta: M \to N$ is a function assigning a module to an agent.
	Finally $\gamma_{i}$ is an \LDLF formula associated to agent $i$.

	\acro{srml} games can be seen as an extension of iBGs in which the strategic power of agents is not apriori fixed but flexibly allocated according to the evolution of the game.
	Typically, modules can be used to enforce or prevent agents to adopt a desired/undesired behaviour.
	As an example, consider the module $\acro{toggle}$ described below:
	
					$$
	\begin{array}{l}
	\MODULE\ \acro{toggle}\ \CONTROLS\ \{p\}\ \\
	
	\quad \INIT \\
	
	\quad \GUARD{}\ \top \ASSIGN\ p := \top; \\
	
	\quad \GUARD{}\ \top \ASSIGN\ p := \top; \\
	
	\quad \UPDATE \\
	
	\quad \GUARD{}\ p \ASSIGN\ p := \bot; \\
	
	\quad \GUARD{}\ \neg p \ASSIGN\ p := \top \\
	
	\end{array}
	$$
	
	~
	
	A player associated to this module is left free to set the value of $p$ on the first round, but then is forced to toggle its value at every round. This kind of constraint on the strategic power of the player associated with agent $\acro{toggle}$ cannot be specified in iBGs. 

	When an agent $i$ is associated to the set of modules $\eta^{-1}(i)$, then, it can control all the variables that are associated to some module in $\eta^{-1}(i)$.
	Moreover, at every step of the execution and for every module, it can use one and only one active guarded command.
	Therefore, every game $\GName$ is equivalent to a game $\GName'$ where every agent $i$ is associated to the module $\bigotimes_{m_i \in \eta^{-1}(i)} m_i$.
	Such translation from $\GName$ to $\GName'$ might produce an exponential blow-up in the representation.
	However, it is not hard to see that all the techniques used in this paper can be easily applied to a \RMGF game with more than a module associated to an agent, thus avoiding this representation expansion and blow-up in the complexity.
	For simplicity, we focus on games of this special form  and denote them by $\GName = \tuple{N, \Phi, m_1, \ldots, m_n, \gamma_{1}, \ldots, \gamma_{n}}$ where $m_i = \tuple{\Phi_i, I_i, U_i}$ is the single module associated to agent $i$, specifying the variable control for the whole set $\Phi_{i}$.

	In this paper, we provide results for both iBGs and RMGs.
	In particular, we show that, for all the problems addressed here, they have the same computational complexity.
	Therefore, given that RMGs are a proper extension of iBGs, we show the upper-bound complexity results for RMGs and lower-bound complexity results for iBGs.
	
	\subsection{Strategies and Nash Equilibria}
	
	Strategies in \iBGF and \RMGF are modelled as {\em deterministic finite state machines}.
	Formally a deterministic finite state machine for player $i$ is a tuple $\sigma_{i} = (S_{i}, s_{i}^{0}, \delta_{i}, \tau_{i})$ where, $S_{i}$ is a finite set of \emph{internal states}, $s_{i}^{0}$ is the \emph{initial state}, $\delta_{i}: S_{i} \times \pow{\APSet} \to S_{i}$ is a \emph{transition function}, and $\tau_{i}: S_{i} \to \pow{\APSet_{i}}$ is the \emph{action function}.
	By $\StrSet_{i}$ we denote the set of possible strategies for player $i$.
	
	In a \RMGF, a strategy $\sigma_{i}$ might not comply with module $m_i$'s specification.
	Hence, for this case we need to define a consistency condition between the module and the strategy.
	We say that $\sigma_{i}$ is \emph{compatible} with module $m_{i}$ if:
	
	\begin{enumerate}
		\item
			$\tau_i(s^0_i) = exec_i(g,\emptyset)$, for some $g \in I_i$; and, 
		\item
			for all $v \in 2^{\Phi}$ and $s, s' \in S_i$, if $s' = \delta_i(s, v)$ then $\tau_i(s') = (\tau_i(s) \setminus ctr(g)) \cup exec_i(g, v)$, for some $g \in U_i$ that is enabled by $v$, {\em i.e.}, such that $v \models guard(g)$, 
	\end{enumerate}
	where $exec_i: (I_i \cup U_i) \times 2^{\Phi} \pto 2^{\Phi_i}$ is the partial function that determines the value of the Boolean variables at the right-hand side of a guarded command when such a guarded command is enabled by a valuation. 
	Formally, $exec_i$ is defined, for a guarded command $g = \phi \ASSIGN x_1'^{i} := \psi_1^{i}; \ \cdots; \ x_k'^{i} := \psi_k^{i}$ and a valuation $v$, as  $exec_i(g,v) = \{x_j^{i} \in \{x_1^{i}, \ldots, x_k^{i} \} \ : \ v \models \psi_j^{i} \}$.
	Intuitively, the compatibility ensures the agent to comply with the module's specification.
	On the one hand, Item 1 states that the guarded command that can be executed on the first state must be an initial command that is enabled on the empty evaluation, that is the evaluation in place before the execution starts.
	On the other hand, Item 2 deals with the fact that, at every iteration, a strategy must execute a guarded command that is enabled in the current state.
	
	From now on, for the case of \RMGF and when it is clear from the context, we will refer to compatible strategies simply as strategies.
	 		
	A (total) \emph{strategy profile} is a tuple $\vec{\sigma} = (\sigma_{1}, \ldots, \sigma_{n})$ of strategies, one for each player.
	We also consider partial strategy profiles.
	For a given set of players $A \subseteq \AgnSet$, we use the notation $\sigma_{A}$ to denote a tuple of strategies, one for each player in $A$.
	Moreover, we use the notation $\sigma_{- A}$ to denote a tuple of strategies, one for each player in $\AgnSet \setminus A$.
	We also use $\sigma_{i}$ in place of $\sigma_{\{ i \}}$ and $\vec{\sigma}_{-i}$ in place of $\vec{\sigma}_{\AgnSet \setminus \{ i \}}$.
	Finally, for two strategy profiles $\vec{\sigma}$ and $\vec{\sigma}'$, by $(\vec{\sigma}_{A}, \vec{\sigma}_{-A}')$ we denote the strategy profile given by associating the strategies in $\vec{\sigma}$ to players in $A$ and strategies in $\vec{\sigma}'$ to players in $\AgnSet \setminus A$.
	
	Since strategies are deterministic, each profile $\vec{\sigma}$ determines a unique infinite play, denoted by $\pi^{\infty}(\vec{\sigma})$, which consists of an infinite sequence of valuations, one for each round of the game.
	Each player $i$ has a preference relation over plays $\pi^{\infty} \in (\pow{\APSet})^{\omega}$, which is determined by its goal $\gamma_{i}$.
	We say that $\pi^{\infty}$ is preferred over ${\pi^{\infty}}'$ by agent $i$, and write $\pi^{\infty} \succeq_{i} {\pi^{\infty}}'$, if and only if ${\pi^{\infty}}' \models \gamma_{i}$ implies that $\pi^{\infty} \models \gamma_{i}$. 
	Using this notion of preference, one can introduce the concept of 
	{\em Nash Equilibrium}.
	We say that $\vec{\sigma}$ is a Nash Equilibrium strategy profile if, for each agent $i$ and
	a strategy $\sigma_{i}' \in \StrSet_{i}$, it holds that $\pi^{\infty}(\vec{\sigma})
	\succeq_{i} \pi^{\infty}(\vec{\sigma}_{-i}, \sigma_{i}')$.
	In addition, by $\NE(\GName) \subseteq \StrSet_{1} \times \ldots \times \StrSet_{n}$ we
	denote the set of Nash Equilibria of the game~$\GName$.
	
	\begin{example}
		\label{exm:filesh-strat}
		Consider again the system in Example~\ref{exm:filesh}.
			A possible strategy $\sigma_{1}$ for player $1$ is a finite-state machine that sets variable $u_{1}$ to true on odd rounds of the execution, while a strategy $\sigma_{2}$ for player $2$ might set $u_{2}$ to true on even rounds of the execution.
		In addition, a possible strategy for player $0$, say $\sigma_{0}$, might be a finite-state machine that sets variable $d_{i}$ to true only after $u_{i}$ has been set to true exactly $n_i$ times in the execution.
		Then, the strategy profile $\vec{\sigma} =(\sigma_{0}, \sigma_{1}, \sigma_{2})$ will be such that the execution $\pi^{\infty} = \pi^{\infty}(\vec{\sigma})$ satisfies $\gamma_{0}$, $\gamma_{1}$, and $\gamma_{2}$, and therefore is a Nash equilibrium.
		Indeed, checking that a strategy profile is a Nash equilibrium of a game is one of the main concerns of this paper, 
		as formalised next.  
	\end{example}
	\paragraph*{Equilibrium Checking}
	We are interested in a number of questions related to the {\em equilibrium analysis} of logic-based multi-player games~\cite{GHW15b,rv-aaai16}.
	
	~
	
	\begin{indent}
		\textbf{\textsc{NE Membership}}.
		Given a game $\GName$ and a strategy profile $\vec{\sigma}$: 
		\begin{center}
			Is it the case
			that $\vec{\sigma} \in \NE(\GName)$?
		\end{center}	
	\end{indent}
	\noindent
	which asks if a strategy profile is a Nash equilibrium of a game.

	The second decision problem we are interested in is the following:

~
	 		
	\begin{indent}
		\textbf{\textsc{NE Non-Emptiness}}. 
		Given a game $\GName$: 
		\begin{center}
			Is it the case that $\NE(\GName) \neq \emptyset$?
		\end{center}		
	\end{indent}
	\noindent 
	which asks if a given game has at least one 
	Nash equilibrium.
	
	Finally, we also consider two decision problems, that are known in the literature as \emph{equilibrium checking}~\cite{rv-aaai16}, formally stated as follows: 
	
~
	
	\begin{indent}
		\textbf{\textsc{E/A-Nash}}.
		Given a game $\GName$ and \LDLF formula $\varphi$:
		\begin{center}
			Does $\pi^{\infty}(\vec{\sigma})
			\models \varphi$ hold, for some/all $\vec{\sigma} \in \NE(\GName)$?
		\end{center}		
	\end{indent}
	\noindent 
	which asks if $\varphi$ is satisfied by some/every
	Nash equilibrium of $\GName$.
	
	In the following sections, we study the above questions for both \iBGF and \RMGF, in particular using an automata-theoretic approach. 

\section{NE Membership}
	\label{sec:membership}
	
	In order to address the \textsc{NE Membership} problem, we first provide some preliminary results on automata.
	An interested reader can find definitions and more details in~\cite{Var95}.
	
	Consider a nondeterministic finite word automaton (NFW) $\AName = \tuple{\Sigma, S, s_{0}, \varrho, F}$, recognizing a regular language $\LName(\AName)$.
	Then consider the nondeterministic B{\"u}chi word automaton (NBW) $\AName[][\infty] = \tuple{\Sigma, S, s_{0}, \varrho', F}$, where, for all $\sigma$ and $s$, we have that 
	$\varrho'(\sigma, s) = \varrho(\sigma, s)$, if $s \notin F$, and 
	$\varrho'(\sigma, s) = \{ s \}$, otherwise.
	
	Intuitively, the automaton $\AName[][\infty]$ mimics the operations of the automaton $\AName$ until an accepting state $s$ is reached.
	From that point on, the automaton disregards mimicking $\AName$ and starts looping indefinitely over $s$.
	Thus, on the one hand, for every finite word $\pi \in \LName(\AName)$, every infinite extension $\pi^{\infty}$, that is, an infinite word such that $\pi$ is a prefix of it, is accepted by $\AName[][\infty]$.
	On the other hand, for an infinite word $\pi^{\infty}$ accepted by $\AName[][\infty]$, there must be a prefix $\pi$ accepted by $\AName$.
	This fact can be shown formally with the following theorem.
	
	\begin{theorem}
		\label{thm:loop}
		Let $\AName$ be a NFW.
		Then, for all $\pi^{\infty} \in \Sigma^{\omega}$, we have that $\pi^{\infty} \in \LName(\AName[][\infty])$ iff there exists $k \in \SetN$ such that $\pi = (\pi^{\infty})_{\leq k} \in \LName(\AName)$.
		In particular, we have that $\LName(\AName[\varphi][\infty]) = \set{\pi^{\infty} \in (\pow{\APSet})^{\omega}}{\pi^{\infty} \models \varphi}$
	\end{theorem}

  \begin{proof}
	The proof proceeds by double implication.
	
	\begin{itemize}
		\item
		Assume $\pi^{\infty} \in \LName(\AName[][\infty])$ and let $\rho^{\infty} \in (2^{S})^{\omega}$ be an accepting run.
		Then, since $\varrho'(\sigma, s) = \{ s \}$ for every $s \in F$, we have that $\rho^{\infty}$ is of the form $s_{0} \cdot s_{1} \cdot \ldots \cdot s_{k - 1} \cdot s_{k}^{\omega}$, with 
		$s_{k} \in F$ and $s_{j} \notin F$ for all $j < k$.
		Then, consider the prefix $\pi = (\pi^{\infty})_{\leq k}$.
		By the definition of $\varrho'$, it follows that $s_{0} \cdot s_{1} \cdot \ldots \cdot s_{k - 1} \cdot s_{k}$ is an accepting run for $\pi$ in $\AName$, and so $\pi \in \LName(\AName)$.
		
		\item
		Let $k \in \SetN$ such that $\pi = (\pi^{\infty})_{\leq k} \in \LName(\AName)$ and let $s_{0} \cdot s_{1} \cdot \ldots \cdot s_{k}$ be an accepting run in $\AName$.
		Moreover, let $j \in \numcc{0}{k}$ be such that $s_{j} \in F$ and $s_{h} \notin F$ for all $h < j$%
		\footnote{Such a $j$ exists because $s_{k}$ is accepting.}
		Then, the infinite sequence $s_{0} \cdot s_{1} \cdot \ldots \cdot s_{j - 1} \cdot s_{j}^{\omega}$ is an accepting run of $\pi^{\infty}$ in $\AName[][\infty]$.
		Thus, we have that $\pi^{\infty} \in \LName(\AName[][\infty])$.
	\end{itemize}
\end{proof}

	To solve the \textsc{NE Membership} problem for \RMGF, we need to account for the fact that strategies adopted by agents must comply with their module specification.
	In terms of automata, we need to make sure that the accepted language is restricted to the infinite play that can be generated by such complying strategies.
	
	Let $\GName = \tuple{N, \Phi, m_1, \ldots, m_n, \gamma_{1}, \ldots, \gamma_{n}}$ be a \RMGF and consider the automaton $\AName[\GName] = \tuple{2^{\Phi}, S_{\GName}, s^{0}_{\GName}, F_{\GName}, \varrho_{\GName}}$ defined as:
	
	\begin{itemize}
		\item
		$S = 2^{\Phi} \cup \{\epsilon\}$;
		
		\item
		$s^{0}_{\GName} = \epsilon$;
		
		\item
		$F_{\GName} = S_{\GName}$;
		
		\item
		$\varrho_{\GName}(\epsilon, v) = 
		\begin{cases}
		v, & \text{ if } \exists g_{1} \in I_1, \ldots g_{n} \in I_n \emptyset \models guard(g_{i}) v = \bigcup_{i \in N}exec_{i}(g_i, \emptyset) \\
		\emptyset, & \text{ otherwise }
		\end{cases}$;

		\item
		$\varrho_{\GName}(s, v) = 
		\begin{cases}
		v, & \text{ if } \exists g_{1} \in U_1, \ldots g_{n} \in U_n \emptyset \models guard(g_{i}) v = \bigcup_{i \in N}exec_{i}(g_i, \emptyset) \\
		\emptyset, & \text{ otherwise }
		\end{cases}$
		
	\end{itemize}
	
	Intuitively, at every round of its execution, the automaton stores the current evaluation of the game in its state $s$.
	Then, when an evaluation $v$ is sent to it, the automaton moves to the state $v$ itself if, and only if, there exists a tuple of commands $g_1, \ldots, g_n$ that are enabled in $s$ and whose combined executions leads the game from $s$ to $v$.
	For the case no sets of enabled commands can lead to $v$, the automaton rejects the path, as it has found an illegal move in it.
	More formally, we have the following result.
	
	\begin{lemma}
		\label{lmm:RMGautomaton}
		Let $\GName$ be a \RMGF and $\AName_{\GName}$ its corresponding automaton.
		Then, it holds that $\LName(\AName_{\GName})$ is exactly the set of possible executions in $\GName$.
	\end{lemma}
	
	\begin{proof}
		We prove the lemma by double inclusion.
		First, assume that $\pi^{\infty}$ is a play in $\GName$ and show that it is accepted by $\AName_{\GName}$.
		Let $s_0, s_1, \ldots$ be the run of $\AName_{\GName}$ over $\pi^{\infty}$.
		By induction on $h$, we prove that $s_{h+1} = \pi^{\infty}_{h}$, for every $h \in \SetN$ and so that the run is accepting.
		As base case, for $h = 0$, we have that $s_1 = \varrho_{\GName}(s_{0}, \pi^{\infty}_{0}) = \varrho_{\GName}(\epsilon, \pi^{\infty}_{0})$ and, since $\pi^{\infty}$ is a legal execution in $\GName$, there exist $g_{1} \in I_1, \ldots, g_{n} \in I_{n}$, with $\emptyset \models guard(g_{i})$, for all $i \in N$, such that $\pi^{\infty}_{1} = \bigcup_{i \in N} exec_{i}(g_i, \emptyset)$, this implying $\varrho_{\GName}(\epsilon, \pi^{\infty}_{0}) = \pi^{\infty}_{0}$, proving the statement for the base case.
		For the induction case, let $h \geq 0$ and assume $s_{h + 1} = \pi^{\infty}_{h}$.
		We have to prove that $s_{h + 2} = \pi^{\infty}_{h + 1}$.
		We have that $s_{h + 2} = \varrho_{\GName}(s_{h + 1}, \pi^{\infty}_{h + 1})$.
		Moreover, since $\pi^{\infty}$ is an execution in $\GName$, there exist $g_{1} \in U_1, \ldots, g_{n} \in U_{n}$, with $\pi^{\infty}_{h} \models guard(g_{i})$, for all $i \in N$, such that $\pi^{\infty}_{h + 1} = \bigcup_{i \in N} exec_{i}(g_i, \emptyset)$.
		Now, observe that by induction hypothesis, we have that $s_{h + 1} = \pi^{\infty}_{h}$, and so it follows that $s_{h + 2} = \varrho_{\GName}(s_{h + 1}, \pi^{\infty}_{h + 1}) = \varrho_{\GName}(\pi^{\infty}_{h}, \pi^{\infty}_{h + 1}) = \pi^{\infty}_{h + 1}$, the last equality following from the definition of $\varrho_{\GName}$.
		
		For the other direction, let us assume $\pi^{\infty} \in \LName(\AName_{\GName})$ and prove that $\pi^{\infty}$ is an execution in $\GName$.
		We have to show prove that there exist $g_{1} \in I_{1}, \ldots, g_{n} \in I_{n}$ such that $\emptyset \models guard(g_{i})$, for all $i \in N$ and $\pi^{\infty}_{0} = \bigcup_{i \in N}exec_{i}(g_{i}, \emptyset)$ and that, for all $h \in \SetN$, there exist $g_{1}^{h} \in U_{1}, \ldots, g_{n}^{h} \in U_{n}$ such that $\pi^{\infty}_{h} \models guard(g_{i}^{h})$, for all $i \in N$ and $\pi^{\infty}_{h + 1} = \bigcup_{i \in N}exec_{i}(g_{i}^{h}, \pi^{\infty}_{h})$.
		But this clearly follows from the definition of $\varrho_{\GName}$ and the fact that the run of $\AName_{\GName}$ over $\pi^{\infty}$ is accepting.
	\end{proof}

	We can now address \textsc{NE Membership}.
	We show that this problem is PSPACE-complete; for the membership argument, we
	employ an automata-based algorithm for checking membership.
	We first introduce, for a given (machine) strategy $\sigma_{i} =
	\tuple{S_{i}, s_{i}^{0}, \delta_{i}, \tau_{i}}$ for a player $i$, a
	corresponding DFW 
	$\AName(\sigma_{i}) =
	\tuple{\Sigma, Q_{i}, q_{i}^{0}, \varrho_{i}, F_{i}}$ where:
	$\Sigma = \pow{\APSet}$ is the alphabet set,
	$Q_{i} = (S_{i} \times \pow{\APSet}) \cup \{ \sink \}$ is
	the state set, where $\sink \notin S_{i} \times \pow{\APSet}$ is a fresh state,
	$q_{i}^{0} = (s_{i}^{0}, \emptyset)$ is the initial state,
	$F_{i} = S_{i} \times \pow{\APSet}$ is the final state set, and 
	$\varrho_{i}$ is the transition relation such that, for all $(s,
	v) \in S_{i} \times \pow{\APSet}$ and $v' \in \Sigma$,
	
	\begin{itemize}
		
		\item
		$\varrho_{i}((s, v), v') = 
		\left\{
		\begin{array}{cl}
		\delta_{i}( s , v ), & \text{ if } \tau_{i}(s) =
		v'|_{\APSet_{i}} \\
		
		\sink, & \text{ otherwise }
		
		\end{array}
		\right.
		$, and
		
		\item
		$\varrho_{i}(\sink, v') = \sink$
		
	\end{itemize}
	
	Let $\LName(\AName(\sigma_{i}))$ denote the set of infinite words in $(\pow{\APSet})^{\omega}$ accepted by $\AName(\sigma_{i})$.
	It is easy to see that such a set is exactly the same set of plays that are
	possible outcomes in a game where player $i$ uses strategy $\sigma_{i}$.
	Similarly, for a given set of players $A \subseteq \AgnSet$ and a partial
	strategy profile $\vec{\sigma}_{A}$, we have that, for $\AName(\vec{\sigma}_{A})
	= \bigotimes_{i \in A}\AName(\sigma_{i})$, the product of these
	automata, the language $\LName(\AName(\vec{\sigma}_{A}))$ contains exactly
	those infinite plays in a game where players in $A$ play according to the
	strategies given in $\vec{\sigma}_{A}$.
	Moreover, in~\cite{DV15} it is shown, for every \LDLF formula $\varphi$, how to build and check on-the-fly a NFW $\AName[\varphi] = \tuple{S, \pow{\APSet}, \{s_{0}\}, \delta, \{s_{f}\}}$, such that, for every finite trace $\pi \in (\pow{\APSet})^{*}$, we have $\pi \models \varphi$ if and only if $\pi \in \LName(\AName[\varphi])$, where by $\LName(\AName[\varphi])$ we denote the language of finite words (that is, the language of finite traces over ${\pow{\APSet}}$) accepted by the automaton $\AName[\varphi]$.
	Such a construction makes use of a function $\delta$ simulating the transition relation of the corresponding 
	alternating finite word automaton (AFW), which takes a subformula $\psi$ of $\varphi$ and a valuation of variables $\Pi \subseteq \APSet$, and recursively returns a combination of subformulae.
	A suitable modification of such an algorithm allows one to construct the NBW $\AName[\varphi][\infty]$.
	As a matter of fact, observe that the only final state $s_{f}$ of the automaton $\AName[\varphi]$ built in \cite{DV15} does not have any outgoing transition.
	Then, given the construction of $\AName[\varphi][\infty]$, we only need to add a loop to it, for every possible valuation.
	
	\begin{algorithm}[h!]
		\caption{\label{alg:ldlnbwprod} Intersection contruction.}
		\textbf{Input}: an \LDLF formula $\varphi$ and an NBW $\AName = \tuple{\pow{\APSet}, Q, q_{0}, \varrho, \F}$.
		
		\textbf{Output}: NBW $\AName[\varphi][\infty] \times \AName = 
		\tuple{\pow{\APSet}, S', \{s_{0}'\}, \F', \varrho'}$.
		
		$s_{0}' \leftarrow \{ (\varphi, q_{0}, 1) \}$ \; 
		$\F' \leftarrow \{ \emptyset \} \times Q \times \{1\}$ \;
		
		$S \leftarrow \{s_{0}\} \cup \F'$ \;
		$\varrho \leftarrow \set{((\emptyset, q, 1), \Pi, (\emptyset, q', 2))}{\Pi \in \pow \APSet \wedge q' \in \varrho(q, \Pi)} \allowbreak \cup \set{((\emptyset, q, 2), \Pi, (\emptyset, q', 2))}{\Pi \in \pow \APSet \wedge q' \in \varrho(q, \Pi) \wedge q \in Q \setminus \F'} \allowbreak \cup \set{((\emptyset, q, 2), \Pi, (\emptyset, q', 1))}{\Pi \in \pow \APSet \wedge q' \in \varrho(q, \Pi) \wedge q \in \F'}$ \;

		\While{($S'$ or $\varrho'$ change)}{
			\For{$s \in S'$, $q \in Q$ and $\Pi \in \pow{\APSet}$}{
				\For{$q' \subseteq CL(\varphi)$ and $q' \in \varrho(q, 		\Pi)$}{
					\If{$s' \models \bigwedge_{\psi \in q}\delta(\psi, \Pi)$}{
						\eIf{$q \in \F'$}{
							$S' \leftarrow S' \cup \{(s, q, 2), (s', q', 1) \}$ \;
							$\varrho' \leftarrow \varrho' \cup \{((s, q, 2), \Pi, (s', q', 1)) \}$\;
						}{
							$S' \leftarrow S' \cup \{(s', q', 1), (s', q', 2), (s, q, 2) \}$ \;
							$\varrho' \leftarrow \varrho' \cup \{((s, q, 2), \Pi, (s', q', 1)) \}$\;
						}
					}
					
				}
			}
			
		}
	\end{algorithm}
	
	However, it cannot be used as it is to obtain the PSPACE complexity for the \textsc{NE Membership} problem.
	Indeed, we need to combine the NBW $\AName[\varphi][\infty]$ with the automata $\AName[\vec{\sigma}]$ and $\AName[\vec{\sigma}_{ - i}]$ provided by the \textsc{NE Membership} problem instance.
	To do this, we need to adapt the construction in order to handle these products.
	Note that both $\AName[\vec{\sigma}]$ and $\AName[\vec{\sigma}_{ - i}]$ can be considered as NBW.
	Thus, it is enough to deliver an algorithm that builds an automaton intersection between $\AName[\varphi][\infty]$ and a generic NBW $\AName$.
	
	\begin{algorithm}[H]
		\caption{\label{alg:RMGmembership} NE Membership for \RMGF.}
		\textbf{Input}: a \RMGF $\GName$ and a strategy profile
		$\strpElm$.
		
		\textbf{Output}: ``Yes'' if $\strpElm \in \NE(\GName)$; ``No'' otherwise.
		
		\If{$\LName(\AName(\strpElm)) \not\subseteq \LName(\AName_{\GName})$}{
			\Return ``Error''
		}
		
		\For{$i \in \AgnSet$}{
			\If{$\LName(\AName(\strpElm) \otimes \AName[\gamma_{i}][\infty]) =
				\emptyset$}{
				\If{$\LName(\AName(\strpElm[-i]) \otimes \AName_{\GName} \otimes \AName[\gamma_{i}][\infty])
					\neq \emptyset$}{
					\Return ``No''
				}
			}
			
		}
		\Return ``Yes''
	\end{algorithm}
	
	When one of the two NBW involved in the product derives from an \LDLF formula $\varphi$, we can adapt the on-the-fly construction provided in~\cite{DV15}, as described in Algorithm~\ref{alg:ldlnbwprod}.
	
	With this construction in place, one can show that Algorithm~\ref{alg:RMGmembership} runs in PSPACE and solves \textsc{NE Membership} for RMG$_F$.
	In particular, the algorithm checks, using the automata constructions presented before, whether a strategy profile is a Nash equilibrium by checking for beneficial deviations in the game for every player.

	\begin{theorem}
		\label{thm:ne-membershipiBG}
		The \textsc{NE Membership} problems for \iBGF and RMG$_{F}$ are PSPACE-complete. 
	\end{theorem}
	
	\begin{proof}
		
		To show that Algorithm~\ref{alg:RMGmembership} is correct, assume that the algorithm returns ``Yes'' on a given instance $(\GName, \vec{\sigma})$.
		This means that it does not return the ``Error'' message in the initial conditional check.
		This means that the outcome of $\vec{\sigma}$ is in the language of $\AName[G]$ and so the strategies are complying with their modules specifications.
		Moreover, the algorithm does not return ``NO'', which means that, for every agent $i$, either the innermost or the outermost conditional checks are false.
		In case the outermost is false, then we have that $\LName(\AName(\vec{\sigma})) \cap \LName(\AName[\gamma_{i}]) \neq \emptyset$, meaning that the play
		$\pi^{\infty}(\vec{\sigma})$ is such that $\pi^{\infty}(\vec{\sigma}) \models \gamma_{i}$.
		Thus, player $i$ is satisfied in the context $\vec{\sigma}$ and so it does
		not have any incentive to deviate from it.
		On the other hand, if the outermost returns true but the innermost returns false, then we have that $\LName(\AName(\vec{\sigma}_{-i})) \cap \LName(\AName[\gamma_{i}]) \cap \LName(\AName[\GName])
		= \emptyset$, which means that the satisfaction of $\gamma_{i}$ is
		incompatible with the partial strategy profile $\vec{\sigma}_{-i}$, no matter how player $i$ behaves compatibly with its modules specification.
		This, in terms of strategies, implies that there is no beneficial deviation for player $i$ to get its goal achieved.
		Hence, the strategy profile $\vec{\sigma}$ is a Nash equilibrium of the game. 
		
		On the other hand, assume $\vec{\sigma}$ is a Nash equilibrium.
		Then, no player $i$ has an incentive to deviate.
		This can be the case for two reasons: either $\pi^{\infty}(\vec{\sigma})
		\models \gamma_{i}$, or there is no compatible strategy $\sigma_{i}'$ such that $\pi^{\infty}(\vec{\sigma}_{-i}, \sigma_{i}') \models \gamma_{i}$.
		If the former, then we have that $\LName(\AName(\vec{\sigma})) \cap
		\LName(\AName[\gamma_{i}]) \neq \emptyset$ and so the check on the outermost conditional is false.
		If the latter, then it follows that $\LName(\AName(\vec{\sigma}_{-i}))
		\cap \LName(\AName[\gamma_{i}]) \cap \LName(\AName[\GName]) = \emptyset$, making the check on innermost conditional is false.
		Since this reasoning holds for every player $i$, it can be concluded that
		Algorithm~\ref{alg:RMGmembership} ends by  returning ``Yes'',  which concludes the proof of correctness. 
		
		Regarding the complexity, note that all the conditional checks involve a nonemptiness test of NFW built by means of the Algorithm~\ref{alg:ldlnbwprod}, whose complexity is PSPACE.
		Since this procedure is called $n$ times, where $n$ is the number of
		agents, we obtain a PSPACE upper bound.
		
		We now show hardness on \iBGF by providing a reduction from the satisfiability problem of
		\LDLF formulae, which is known to be PSPACE-complete~\cite{DV13}.
		Consider an \LDLF formula $\varphi$ and then define the one-player game
		$\GName$, with player set $\{ 1 \}$, in which player $1$ controls all the
		variables in $\varphi$ plus an additional variable $\{ p \}$ that does
		not appear in $\varphi$, and whose goal is $\gamma_1 = \varphi \wedge
		p$. 
		Moreover, let $\sigma$ be the strategy for player $1$ that myopically plays
		$\emptyset$ in all rounds. Then, such a strategy, which clearly is linear in
		the size of $\varphi$ since it has constant size, is such that $\sigma
		\models \neg \gamma_1$. 
		Now, we will show based on this reduction that $\varphi$ is satisfiable if
		and only if $\sigma \not \in \NE(\GName)$.
		Firstly, if $\varphi$ is satisfiable then player $1$ can (beneficially)
		deviate to a myopic strategy, say $\sigma'$, that generates an infinite play
		with $p$ in the first round and such that one of its prefixes satisfies
		$\varphi$---in which case $\sigma' \models \gamma_1$.
		Then, $\sigma \not \in \NE(\GName)$. 
		On the other hand, if $\varphi$ is not satisfiable, then $\gamma_1$ is not
		satisfiable either. 
		Then, it is clear that there is no strategy $\sigma'$ to which player $1$ 
		can beneficially deviate to achieve its goal; hence $\sigma$ is a Nash
		equilibrium of $\GName$. 
		Because PSPACE is closed under complement, PSPACE-hardness follows. 
	\end{proof}

	\section{NE Non-Emptiness and Equilibrium Checking Problems}
	
	Now, let us study \textsc{NE Non-Emptiness} for both \iBGF and \RMGF.
	We first prove a result for \iBGF and then how to adapt it for the case of \RMGF. 
	Also in this case we use an automata-theoretic approach.
	We show how, given a game $\GName$, it is possible to construct an alternating
	automaton $\AName[\NE](\GName)$ such that $\AName[\NE](\GName)$ accepts
	precisely the set of plays that are generated by the Nash equilibria of
	$\GName$. 
	A distinguishing feature of our automata technique is that it is
	\emph{language preserving}, that is, $\AName[\NE](\GName)$ recognizes exactly
	the set of plays that are obtained by some Nash equilibrium in the game.
	Hereafter, we call {\em Nash runs} the elements in such a set of runs.
	This property of our construction is the key to show that the set of Nash 
	runs is, in fact, $\omega$-regular.
	Also, note that as we now have to find (and not simply check) a strategy
	profile, we cannot use the automata of the form $\AName(\sigma_{i})$ provided
	above, as there is no known strategy $\sigma_{i}$, for each player
	$i$, that can be used here. 
	
	First, we recall the characterisation of Nash equilibria provided in~\cite{GHW15a}.
	For a given \RMGF $\GName = \tuple{N, \Phi, m_1, \ldots, m_n, \gamma_{1}, \ldots, \gamma_{n}}$ and a designated player $j \in \AgnSet$, we say that $\vec{\sigma}_{-j}$ is a punishment profile against $j$ if, for every strategy $\sigma_{j}'$, it holds that $(\vec{\sigma}_{-j}, \sigma_{j}') \not\models \gamma_{j}$.
	In~\cite{GHW15a}, it has been proven that $\vec{\sigma} \in \NE(\GName)$ if and only if there exists $W \subseteq \AgnSet$ such that $\sigma \models \gamma_{i}$ for every $i \in W$ and, for every $j \in L = \AgnSet \setminus W$, the profile $\vec{\sigma}_{-j}$ is a punishment strategy against $j$, that is, a winning strategy profile of the coalition of players $\AgnSet_{-j}$ for the negation of the goal of player $j$. 
	
	Thus, we can think of finding punishment strategies in terms of synthesizing a finite state machine controlling $\APSet_{-j}$.
	To do this, we apply an automata-theoretic approach.
	First of all, we build the alternating Rabin word automaton (ARW) $\AName[\gamma_{j}]$, used to recognize the models of $\gamma_{j}$, and then the product $\AName[\gamma_{j}][\GName] = \AName[\gamma_{j}] \oplus \AName[\GName]$, filtering the models that are compatible with an execution of the RMG$_{F}$ $\GName$.
	Analogously, the automaton $\AName_{\overline{\gamma_{i}}}^{\GName} = \overline{\AName_{\gamma_{i}}} \otimes \AName_{\GName}$ recognizes the plays that both are compatible with the \RMGF $\GName$ and do not satisfy $\gamma_{i}$.
	At this point, by means of Theorem 2 in~\cite{PR89}, we build a nondeterministic Rabin automaton on trees (NRT) $\overline{\AName[\gamma_{j}][\GName]'}$ that recognizes exactly those trees $T$ that are obtained from an execution of a compatible winning strategy of the coalition $\AgnSet_{-j}$ when the goal is to avoid the satisfaction of $\gamma_{j}$.
	Now, following Corollary 17 in~\cite{NW98}, we can build a NRW $\overline{\AName[\gamma_{j}][\GName]''}$ such that $\LName(\overline{\AName[\gamma_{j}][\GName]''}) = \set{\pi^{\infty} \in (\pow{\APSet})^{\omega}}{\exists T \in \LName(\overline{\AName[\gamma_{j}][\GName]'}). \ \pi^{\infty} \subseteq T}$, where by $\pi^{\infty} \subseteq T$ we denote the fact that $\pi^{\infty}$ is a branch of the tree $T$ starting at its root.
	
	Now, let us fix $W \subseteq \AgnSet$ for a moment, and consider the
	product automaton $\AName[\overline{L}] = \bigotimes_{j \in L}
	\overline{\AName[\gamma_{j}][\GName]''}$.
	By the semantics of the product operation we obtain that $\AName[\overline{L}][\GName]$ accepts those paths that are generated by some punishment profile, compatible with $\GName$, for each $j \in L$.
	Moreover, consider the automaton $\AName_{W} = \bigotimes_{i \in W} \AName[\gamma_{i}][\GName]$, recognizing the paths that satisfy every $\gamma_{i}$, for $i \in W$.
	Thus, we have that the product automaton $\AName_{W} \otimes \AName[\overline{L}]$ accepts exactly those paths for which every $\gamma_{i}$, with $i \in W$, is satisfied while, for each $j \in L$ the coalition $\AgnSet_{-j}$ is using a punishment strategy against $j$.
	Now, in order to exploit the characterisation given in~\cite{GHW15a}, we only need to quantify over $W \subseteq \AgnSet$.
	This, in terms of automata, corresponds to the union operation.
	Then, we get the following automata characterisation:
	
	\begin{center}
		$\AName[\NE](\GName) = \bigoplus_{W \subseteq \AgnSet} 
		(\AName_{W} \otimes \AName[\overline{L}]$).
	\end{center}
	
	\begin{theorem}[\RMGF Expressiveness]
		\label{thm:NEreg}
		For a \RMGF game $\GName$, the automaton $\AName[\NE](\GName)$ recognizes the set of Nash runs of $\GName$.
		Therefore, the set of Nash runs of $\GName$ is $\omega$-regular. 
	\end{theorem}
	
	\begin{proof}
		We prove the theorem by double implication.
		From left to right, assume that $\pi^{\infty} \in \LName(\AName[\NE](\GName))$.
		Then, there is $W \subseteq \AgnSet$ such that $\pi^{\infty} \in
		\LName(\AName_{W} \otimes \AName[\overline{L}])$.
		Observe that, w.l.o.g.\ we can assume that $\pi^{\infty}$ is an ultimately periodic play~\cite{CS85} and so that there exists a finite-state machine
		$\Delta_{\pi^{\infty}} = (Q_{\pi^{\infty}}, q_{\pi^{\infty}}^{0},
		\delta_{\pi^{\infty}}, \tau_{\pi^{\infty}})$, controlling all the variables in
		$\APSet$, {\em i.e.}, $\tau_{\pi^{\infty}}: Q_{\pi^{\infty}} \to \pow{\APSet}$, that
		generates $\pi^{\infty}$.
		Moreover, observe that, for each $j \in L$, $\pi^{\infty} \in
		\LName(\overline{\AName[\gamma_{j}][\GName]''})$ implies that there exists
		$T_{j} \in \LName(\overline{\AName[\gamma_{j}][\GName]'})$ such that $\pi^{\infty}
		\subseteq T_{j}$.
		This implies that, for each $j \in L$, there is a finite-state
		machine $\Delta_{j} = (Q_{j}, q_{j}^{0}, \delta_{j}, \tau_{j})$,
		controlling all the variables but $\APSet_{j}$, {\em i.e.}, $\tau_{j}: Q_{j} \to
		\pow{\APSet[ - j]}$, that generates the branches of $T_{j}$, according to
		the output of variables in $\APSet_{j}$, including $\pi^{\infty}$.
		Now, for each $i \in \AgnSet$, define the strategy
		$\sigma_{i} = (S_{i}, s_{i}^{0}, \delta_{i}, \tau_{i})$ as follows:
		
		\begin{itemize}
			\item
			$S_{i} = Q_{\pi^{\infty}} \times \bigtimes_{j \in L} Q_{j}
			\times (L \cup \{ \top \})$ is the product of the state-space of
			$\Delta_{\pi^{\infty}}$ together with the state-space of each 
			$\Delta_{j}$, for each $j \in L$, plus a flag component
			given by $L \cup \{ \top \}$;
			
			\item
			$s_{i}^{0} = (q_{\pi^{\infty}}^{0}, q_{j_{1}}^{0}, \ldots,
			q_{j_{\card{L}}}^{0}, \top)$, collecting all the initial states of
			the finite state machines, $\Delta_{\pi^{\infty}}$ and $\Delta_{j}$, for
			each $j \in L$, flagged with the symbol $\top$;
			
			\item
			$\delta_{i}$ is defined as follows: for each $(q, q_{j_{1}}, \ldots, q_{\card{L}}, \top)$ and $v \in \pow{\APSet}$, $\delta_{i}((q, q_{j_{1}}, \allowbreak \ldots, q_{\card{L}}, \top), v) = (\delta_{\pi^{\infty}}(q, v), \delta_{j_{1}}(q_{j_{1}}, v), \allowbreak \ldots, \delta_{j_{\card{L}}}(q_{\card{L}}, v), \mthfun{flag})$, where $\mthfun{flat} = \top$ if $v = \tau_{\pi^{\infty}}(q)$ and $\mthfun{flat} = j$ if $v_{ - j} = (\tau_{\pi^{\infty}}(q))_{-j}$ and $v_{j} \neq (\tau_{\pi^{\infty}}(q))_{j}$.
			
			\item
			$\tau_{i}((q, q_{j_{1}}, \ldots, q_{\card{L}}, \top)) = 
			(\tau_{\pi^{\infty}}(q))_{i}$ and $\tau_{i}((q, q_{j_{1}},
			\ldots, q_{\card{L}}, \allowbreak j)) = (\tau_{j}(q))_{i}$, for each $j 
			\in L$.
		\end{itemize}
		
		Intuitively, a strategy $\sigma_{i}$ for player $i$ runs in parallel the $i$-th component of the finite-state machine $\Delta_{\pi^{\infty}}$ together with the $i$-th components of the finite-state machines $\Delta_{j}$ that win against the deviating players in $L$.
		Note that, by construction, as long as nobody deviates, the outcome of every single $\Delta_{j}$ corresponds to the one of $\Delta_{\pi^{\infty}}$.
		We have that the strategy profile $\vec{\sigma}$, given by the union of the strategies defined above, generates $\pi^{\infty}$, and, as soon as a unilateral deviation occurs from player $j \in L$, the partial strategy profile $\vec{\sigma}_{-j}$ starts following the finite-state machine $\Delta_{j}$, which is by definition winning against~$j$.
		Thus, $\vec{\sigma}$ is a Nash equilibrium.
		
		From right to left, assume that $\pi^{\infty}$ is a Nash run and let $\vec{\sigma}$
		be a Nash equilibrium such that $\pi^{\infty}(\vec{\sigma}) = \pi^{\infty}$.
		Moreover, let $W = \set{i \in \AgnSet}{\pi^{\infty} \models \gamma_{i}}$.
		We show that $\pi^{\infty} \in \LName(\AName_{W} \otimes
		\AName[\overline{L}])$.
		Since $\pi^{\infty} \models \gamma_{i}$, for each $i \in W$, we have that
		$\pi^{\infty} \in \LName(\AName_{W})$.
		Moreover, let $j \in L$.
		It holds that $j$ does not have a beneficial deviation from $\vec{\sigma}$ and
		so we have that $\vec{\sigma}_{-j}$ is a winning strategy against $j$.
		From the definition of $\overline{\AName[\gamma_{j}][\GName]'}$ we have that the
		tree-execution $T_{j}$ generated by $\vec{\sigma}_{-j}$ is in
		$\LName(\overline{\AName[\gamma_{j}][\GName]'})$.
		Now, since $\pi^{\infty} \subseteq T_{-j}$, we have that $\pi^{\infty} \in
		\LName(\overline{\AName[\gamma_{j}][\GName]''})$, for each $j \in L$, implying
		that $\pi^{\infty} \in \LName(\AName[\overline{L}])$.
		Hence, we have that $\pi^{\infty} \in
		\LName(\AName_{W}) \cap \LName(\AName[\overline{L}]) =
		\LName(\AName_{W} \otimes \AName[\overline{L}])$, as required. 
	\end{proof}

	Using Theorem~\ref{thm:NEreg} we can address the problem of deciding
	if a game admits a Nash equilibrium by checking 
	$\AName[\NE](\GName)$ for emptiness. 
	Regarding the complexity of building $\AName[\NE](\GName)$, 
	observe that the construction of each
	automaton $\overline{\AName[\gamma_{j}][\GName]'}$, provided in~\cite{PR89}, is of
	size doubly exponential with respect to $\card{\gamma_{j}}$.
	Moreover, all the other operations used to build $\AName[\NE](\GName)$
	involve union and intersection of Rabin automata, which can be performed in
	time polynomial in the size of the constituting components.
	This shows that $\AName[\NE](\GName)$ is a nondeterministic Rabin automaton on
	words of size doubly exponential with respect to the game $\GName$.
	Since checking emptiness of a NRW can be done in NLOGSPACE, we obtain the following result. 
	
	\begin{theorem}
		\label{thm:ldlNEnonempt}
		\textsc{NE Non-Emptiness} of \RMGF can be solved in 2EXPTIME. 
	\end{theorem}
	
	Now, to show that \textsc{E-Nash} and \textsc{A-Nash} are in
	2EXPTIME, we can also apply an automata-theoretic approach.
	Indeed, for the \textsc{E-Nash} case, consider a game $\GName$ and an \LDLF
	formula $\varphi$.
	Then, the automaton $\AName[\varphi] \otimes \AName[\NE](\GName)$ recognizes
	all the plays that both satisfy $\varphi$ and are a Nash run.
	Thus, checking the \textsc{E-Nash} problem corresponds to checking the
	nonemptiness of such automaton.
	On the other hand, for the \textsc{A-Nash} problem, consider the automaton
	$\overline{\AName[\varphi]} \otimes \AName[\NE](\GName)$.
	This product automaton recognizes all plays that do not satisfy the formula 
	$\varphi$ and are a Nash run.
	Thus, checking the \textsc{A-Nash} problem corresponds to checking the 
	emptiness of such an automaton.
	The two constructions above show that both \textsc{E-Nash} and \textsc{A-Nash}
	can be solved in 2EXPTIME.
	Formally, combining the results above, we also obtain the following theorem: 
	
	\begin{theorem}
		\label{thm:ldlEnashAnash}
		\textsc{E-Nash} and \textsc{A-Nash} for \RMGF can be
		solved in 2EXPTIME. 
	\end{theorem}

\section{Extensions and Restrictions}
	\label{sec:IV}
	
	We now investigate on some extensions and restrictions on the problems studied in the previous section. 
	As a first result, we show that an extension of the \LDLF language used to represent players' goals can be used to encode \LDLF synthesis, studied in~\cite{DV15}, as a \textsc{NE Non-Emptiness} problem.
	Subsequently, we restrict to two classes of strategies, namely \emph{memoryless} and \emph{myopic} strategies.
	With respect to memoryless strategies, we show that our automata-based
	techniques can be used to show that the set of Nash runs for games of this 
	kind is also $\omega$-regular, as in the original problem.
	An EXPSPACE brute-force approach can be used to show that the induced automata 
	are suboptimal from a complexity point of view.%
	\footnote{Of course, with respect to \textsc{Expressiveness} this is an irrelevant feature of the automata construction.}
	However, the construction is still based on a simple extension of 
	automata on finite words, making it potentially useful in practice. 
	The case of myopic strategies, instead, is studied using a reduction
	to the satisfiability problem for the $1$-alternation fragment of \QPTL, 
	known to be solvable in EXPSPACE~\cite{SVW87}. 
	
	\begin{paragraph}{Games with Quantified prefix \LDLF Goals} 
		
		The results obtained so far show that checking whether a game has a Nash
		equilibrium can be solved in 2EXPTIME.
		We now show that an extension of the logic \LDLF, which we call \emph{quantified prefix \LDLF} (\QPLDLF) can also be solved using the same automata-theoretic technique, with the same complexity, and can be used to represent the \LDLF synthesis problem, which is 2EXPTIME-complete.
		Then, \textsc{NE Non-Emptiness} with respect to such an extension is 2EXPTIME-complete.
		
		Syntactically, a \QPLDLF formula $\varphi$ is obtained from an \LDLF formula $\psi$ by simply adding either an existential $\exists$ or a universal $\forall$ quantifier in front of it, i.e., $\varphi = \exists \psi$ or $\varphi = \forall \psi$.
		Such a quantification ranges over the set of prefixes of a given infinite path of valuations.
		Formally, we have that, for a given \QPLDLF formula of the form $\exists \psi$ and an infinite path $\pi^{\infty}$, we have that $\pi^{\infty} \models \exists \psi$ if there is $k \in \SetN$ such that $\pi^{\infty}_{ < k} \models \psi$.
		Analogously, for a \QPLDLF formula of the form $\forall \psi$, we have $\pi^{\infty} \models \forall \psi$ if $\pi^{\infty}_{ < k} \models \psi$, for all $k \in \SetN$.
		
		The reader might note that $\exists \psi$ is equivalent to $\psi$ on infinite
		plays.
		This means that the set of models for $\exists \psi$ corresponds to the set
		of infinite models of $\psi$ and so the automaton $\AName[\exists \psi] =
		\AName[\psi]$ recognizes the models of $\exists \psi$.
		Moreover, observe that, for every \LDLF formula $\psi$ and an infinite play
		$\pi^{\infty}$, we have that $\pi^{\infty} \models \forall \psi$ iff $\pi^{\infty}
		\not\models \exists \neg \psi$.
		This means that, in order to build the automaton $\AName[\forall \psi]$ for
		a formula of the form $\forall \psi$, one can first consider the formula $\neg \psi$
		and build the corresponding automaton $\AName[\exists \neg \psi]$.
		It follows that $\LName(\AName[\exists \neg \psi])$ is the set of infinite 
		plays that satisfy $\exists \neg \psi$, which is the complement of the set
		of plays satisfying $\forall \psi$.
		Thus, $\AName[\forall \psi] = \overline{\AName[\exists \neg \psi]}$.
		Using these constructions one can solve 
		\textsc{NE Non-Emptiness}, \textsc{E-Nash}, and \textsc{A-Nash} with \QPLDLF
		goals by applying the same automata-theoretic technique used for \LDLF.
		Then, we have the following result. 
		
		\begin{theorem}
			\label{thm:qpldl_upper}
			\textsc{NE Non-Emptiness}, \textsc{E-Nash}, and \textsc{A-Nash} with \QPLDLF goals, for both \iBG and \RMG can be solved in 2EXPTIME.
			Moreover, the sets of Nash equilibria for these classes of games is $\omega$-regular. 
		\end{theorem}
		
		To obtain a matching lower bound, observe that, given the interpretation
		of \QPLDLF formulae, it is possible to encode the synthesis problem for
		\LDLF formulae as presented in~\cite{DV15}.
		Indeed, in such a case we only have to set a two-player game $\GName$ in
		which, say Player~1, controls the same variable as the system for the
		synthesis problem, and Player~2 controls the environment variables.
		At this point, by setting $\gamma_{1} = \exists \psi$ and $\gamma_{2} =
		\forall \neg \psi$, one ensures that Player~1 and Player~2 have exactly the
		same behaviours of system and environment in the synthesis problem,
		respectively.
		In addition to Player~1 and Player~2, to ensure a reduction to 
		\textsc{NE Non-Emptiness} one can add two players that
		trigger a ``matching pennies'' game in case $\psi$ is not
		synthesised.
		With this reduction it follows that \textsc{NE Non-Emptiness} is
		2EXPTIME-complete.
		
		Formally, consider an \LDLF formula $\varphi$ and the synthesis problem for
		it, in which the system controls a set of (output) variables $X$ while the
		environment controls a set of (input) variables $Y$.
		Then, consider the four-player \iBGF $\GName[\varphi]$ with \QPLDLF goals such that:
		
		\begin{itemize}
			\item
			Player 1 controls $X$ and has $\gamma_{1} = \exists \varphi$ as
			goal;
			
			\item
			Player 2 controls $Y$ and has $\gamma_{2} = \forall \neg \varphi$ as
			goal;
			
			\item
			Player 3 controls a fresh Boolean variable $p$ and has $\gamma_{3} = \exists
			\varphi \vee (p \leftrightarrow q)$ as goal; and
			
			\item
			Player 4 controls a fresh Boolean variable $q$ and has $\gamma_{4} = \exists
			\varphi \vee \neg (p \leftrightarrow q)$ as goal.
		\end{itemize}
		
		Using the above construction, we can show that the synthesis problem for an \LDLF formula $\varphi$ can be solved by addressing the \textsc{NE Non-Emptiness} problem for $\GName[\varphi]$, from which we derive the following theorem.
		
		\begin{theorem}
			\label{thr:qpldl_lower}
			\textsc{NE Non-Emptiness}, \textsc{E-Nash}, and \textsc{A-Nash} are 2EXPTIME-complete for both \iBGF and \RMGF.
		\end{theorem}
		
		In fact, Theorem~\ref{thr:qpldl_lower} is proved using the lemma given below. 
		
		\begin{lemma}
			\label{lmm:synred}
			The synthesis problem for an \LDLF formula $\varphi$ over a set of Boolean variables $X \cup Y$, where the system controls the variables in $X$ and the environment the variables in $Y$ has a positive answer if and only if the game $\GName[\varphi]$ has a Nash equilibrium.
			
		\end{lemma}
		
		\begin{proof}
			
			We prove the lemma for \iBGF by double implication.
			From left to right, assume that the system has a winning strategy
			$\sigma_{1}$ against the environment in the synthesis problem.
			Then every strategy profile $\vec{\sigma}$ in $\GName[\varphi]$ in which
			Player 1 uses $\sigma_{1}$ is a Nash equilibrium.
			Indeed, as $\sigma_{1}$ is a winning strategy for the synthesis problem,
			we have that $\vec{\sigma} \models \varphi$ and so $\vec{\sigma} \models \exists
			\varphi$.
			Then, Players 1, 3, and 4 have their goals satisfied and,
			therefore, do not have an incentive to deviate.
			On the contrary, Player 2 does not have its goal $\gamma_{2}$ satisfied. 
			However, there is no beneficial deviation $\sigma_{2}'$ such that
			$(\vec{\sigma}_{-2}, \sigma_{2}') \models \gamma_{2}$, otherwise,
			$\sigma_{1}$ would not be a winning strategy in the synthesis problem.
			
			From right to left, assume there exists a strategy profile $\vec{\sigma}$ that
			is a Nash equilibrium.
			It is not hard to see that we have that $\vec{\sigma} \models \exists
			\varphi$, otherwise either Player 3 or Player 4 has a beneficial
			deviation.
			Moreover, the corresponding strategy $\sigma_{1}$ for Player 1 is winning
			for the system in the synthesis problem.
			Indeed, if by contradiction there exists a strategy $\sigma_{2}'$ for
			the environment in the synthesis problem, then we have that $(\vec{\sigma}_{-2}, \sigma_{2}') \models \gamma_{2}$, and so $\sigma_{2}'$ is a
			beneficial deviation for Player 2 in $\GName[\varphi]$, which 
			contradicts the fact that $\vec{\sigma}$ is a Nash equilibrium.
			
			For the \RMGF case, we can just regard the \iBGF $\GName[\varphi]$ as a \RMGF in which every agent $i$ is associated with $\card{\Phi_{i}}$ modules, each of them allowing a single variable $x_i \in \Phi_{i}$ to be freely set at every iteration of the game execution.
			Thus, we obtain the assert. 
		\end{proof}
		
		Theorem~\ref{thr:qpldl_lower} is a direct consequence of Lemma~\ref{lmm:synred}.
		
	\end{paragraph}
	
	\begin{paragraph}{Games with Memoryless Strategies}
		\label{secn:memoryless}
		
		In this subsection, we study games with memoryless strategies.
		We say that a strategy $\sigma_{i} = (S_{i}, s_{i}^{0}, \delta_{i},
		\tau_{i})$ for Player $i$ is \emph{memoryless} if $S_{i} = \pow{\APSet}$
		and $\delta_{i}$ is deterministic. 
		Intuitively, a strategy is memoryless if, for each state of the
		game, it always chooses the same action at such state. 
		Moreover, a play $\pi^{\infty} \in (\pow{\APSet})^{\omega}$ is said
		to be memoryless if, for all $v, w \in \pow{\APSet}$, if
		$\pi^{\infty}_{k} = v$ and $\pi^{\infty}_{k + 1} = w$, for some $k \in
		\SetN$, then, for all $h \in \SetN$, if $\pi^{\infty}_{h} = v$ then
		$\pi^{\infty}_{h + 1} = w$.
		A profile $\vec{\sigma}$ made by memoryless strategies can
		only generate memoryless plays and vice-versa.
		
		Moreover, it is not hard to build a polynomial size NBW $\AName[mless]$ accepting all and only the memoryless plays.
		This turns out to be useful in addressing the case of memoryless strategies.	
		Indeed, to solve the \textsc{NE Non-Emptiness} problem with memoryless strategies, 
		we only need to adjust the general procedure by pairing the automaton $\AName[mless]$ to every single component of the automaton $\AName[NE](\GName)$.
		This operation then adds the memoryless requirement to the goal of
		a player and to the punishment strategies.
		
		Now, although this solution technique allows one to prove that the set of Nash Equilibria in memoryless games is $\omega$-regular, this is not optimal from a computational complexity point of view, which is still 2EXPTIME.
		For instance, a brute-force procedure can solve the problem in EXPSPACE.
		Indeed, given the definition of strategies, we know that a memoryless strategy for a player in the game has (at most) $2^\Phi$ states.
		Then, a memoryless strategy, as well as a strategy profile, can be guessed
		in time exponential in the size of $\Phi$ and saved using exponential space. 
		In addition, using \textsc{NE Membership} we can check in PSPACE whether such a strategy profile is a Nash equilibrium of the game.
		Moreover, for some \RMGF the problem can be solved even with a better complexity.
		Indeed, if the overall number of guards in the modules is polynomial w.r.t the number of variables, we only have a polynomial number of states, and thus memoryless strategies can be polynomially represented, leading to a PSPACE complexity for the \textsc{NE Non-Emptiness} problem.
		Formally, we have:

		\begin{theorem}
			The sets of memoryless Nash equilibria for \iBGF and \RMGF with both \LDLF and \QPLDLF are $\omega$-regular. 
			Moreover, the \textsc{NE Non-Emptiness} problem for \iBGF and \RMGF with both \LDLF and \QPLDLF with memoryless strategies can be solved in EXPSPACE.
			For the case of \RMGF with an overall number of guards that is polynomial w.r.t. the number of variables, the problem can be solved in PSPACE.
		\end{theorem}

	\end{paragraph}

	\begin{paragraph}{Games with Myopic Strategies}
		\label{secn:myopic}
		
		Another important game-theoretic setting is the one given by 
		{\em myopic strategies} 
		as they can be used to define {\em all} beneficial deviations. 
		A game with myopic strategies is called a {\em myopic \iBGF}.
		We say that a strategy $\sigma_{i} = (S_{i}, s_{i}^{0}, \delta_{i},
		\tau_{i})$ for Player $i$ is myopic if its transition function does not
		depend on the input variables, i.e., such that for each $s \in S_{i}$
		and $v, v' \in \pow{\APSet}$, we have $\delta_{i}(s,
		v) = \delta_{i}(s, v')$.
		In a myopic \iBGF, players are only allowed to use myopic strategies.
		In~\cite{GPW16} it is shown how to 
		reduce
		\textsc{NE Non-Emptiness} for myopic iBG to the satisfiability of the
		\QPTL formula 
		$$
		\varphi = \bigvee_{W \subseteq \AgnSet}  (\exists \Phi_1, \ldots,
		\Phi_n . (\bigwedge_{i \in W} \gamma_i \wedge \bigwedge_{j \in
			\AgnSet \setminus W} (\forall\Phi_j . \neg\gamma_j)))
		$$
		where the formulae $\gamma_{i}$ are the \LTL goals of the players in the myopic iBG instance and the quantifier alternation is $1$ (an alternation fragment for which the complexity is known to be EXPSPACE~\cite{SVW87}).
		
		To apply the solution provided in~\cite{SVW87} to check the satisfiability of $\varphi$, one first has to transform each $\gamma_{i}$ into the NBW automata recognizing their models. 
		In the case of \iBGF, these \LTL formulae are replaced by \LDLF formulae.
		However, as shown in the previous section, the infinite models of an \LDLF formula $\gamma_{i}$ can also be recognized by NBW automata that are equivalent to some $\omega$-regular expression of the form $\alpha \cdot (\pow{\APSet})^{\omega}$.
		Thus, in order to solve \textsc{NE Non-Emptiness} for myopic \iBGF, we can first transform every \LDLF goal $\gamma_{i}$ into the corresponding NBW $\AName[\gamma_{i}]$ and then follow the technique used in~\cite{SVW87}.
		Note that the same reasoning applies also for the case of \QPLDLF goals.

		Moreover, for the case of \RMGF, we just have to replace the automaton $\AName_{\gamma_{i}}$ for $\gamma_{i}$ with the automaton $\AName_{\gamma_{i}} \otimes \AName_{\GName}$, to force every player to comply with the modules specification.

		We then obtain the following result for games with myopic strategies: 
		
		\begin{theorem}
			The \textsc{NE Non-Emptiness} problem for myopic \iBGF and \RMGF with \LDLF or \QPLDLF goals can be solved in EXPSPACE.
		\end{theorem}
		At this point it is important to note that a key observation behind this result is the fact that when playing with myopic strategies the strategies that are used to construct a run that is sustained by a Nash equilibrium (a Nash run) must be oblivious to players' deviations.
		
	\end{paragraph}

	\begin{paragraph}{Games with Strong Nash equilibria}
		\label{secn:strongne}
		
		Despite being the most used solution concept in non-cooperative game
		theory~\cite{OR94}, Nash equilibrium still has some limitations, for instance,
		it is not always stable and also it includes non desirable
		equilibria.
		As an example, consider a two-player game in which Player 1 controls a
		variable $p$ and has the \LDLF goal $\gamma_{1} = q$, while Player
		2 controls a variable $q$ and has the  \LDLF goal $\gamma_{2} =
		p$%
		\footnote{Observe that $\gamma_{1}$ and $\gamma_{2}$ are
			propositional logic formulae, {\em i.e.}, special cases of \LDLF.}.
		It is clear that every strategy profile $\vec{\sigma}$ is a Nash
		equilibrium.
		Indeed, even in case a goal $\gamma_{i}$ is not satisfied, the corresponding
		player cannot deviate from it, as the satisfaction of each player's goal is
		fully controlled by the other one.
		However, the desired outcome for both players is to satisfy both goals.
		Then, if we allow the two players to \emph{collaboratively} deviate, the only
		stable outcomes are the ones making true both $p$ and $q$ at the first
		round of the computation. 
		
		A strong Nash equilibrium considers not only a single player's deviation, but
		also every possible coalition of players having a collective deviation
		incentive.
		Formally, for a given strategy profile $\vec{\sigma}$, we say that it is a
		\emph{strong Nash equilibrium} if there is no subset $C 
		\subseteq \AgnSet$ and partial strategy profile $\vec{\sigma}_{C}'$ such
		that, for all $i \in C$, $\pi^{\infty}(\vec{\sigma}_{-C},
		\vec{\sigma}_{C}') \succ_{i} \pi^{\infty}(\vec{\sigma})$.
		Then, in a strong Nash equilibrium a coalition of players~$C$ has an
		incentive to deviate if and only if every player $i$ in such a
		coalition has an incentive to deviate.
		By $\sNE(\GName)$ we denote the set of strong Nash equilibria in $\GName$.
		To check whether there exists a strong Nash equilibrium in a game
		using an automata-theoretic approach, we need to be able to express
		this notion of beneficial collective deviation with an appropriate
		automaton.
		
		To do this, we just need to adjust the automaton
		$\overline{\AName[\gamma_{j}][\GName]''}$ given in the previous section,
		used to recognize all the plays that can be generated by a punishment
		strategy of the coalition $\AgnSet[ - j]$ against $j$, having goal
		$\gamma_{j}$.
		Indeed, the concept of punishment can be easily lifted to punishing a group of
		players.
		To do this, for a set $C \subseteq \AgnSet$, consider the automaton
		$\AName_{C} = \bigotimes_{j \in C} \AName[\gamma_{j}][\GName]$ recognizing
		all the models that satisfy every $\gamma_{j}$, for $j \in C$.
		Then, as in the previous section, we can build the automaton
		$\overline{\AName_{C}''}$ that recognizes the plays generated by a
		punishment strategy for the coalition $\AgnSet \setminus C$ against the
		goal being the conjunction of goals of coalition $C$.
		At this point, as in the case of Nash equilibrium, let us fix a set of
		``winners'' $W \subseteq \AgnSet$ in the game and then consider the product automaton
		$\AName[\overline{\pow{L}}] = \bigotimes_{C \in \pow{L}}
		\overline{\AName_{C}''}$.
		By the semantics of the product operation we obtain that the automaton 
		$\AName[\overline{\pow{L}}]$ accepts those paths that are generated by
		some punishment profile, for each coalition of players $C \in \pow{L}$.
		
		Thus, we have that the product automaton $\AName_{W} \otimes
		\AName[\overline{\pow{L}}]$ accepts exactly those paths for which every
		$\gamma_{i}$, with $i \in W$, is satisfied while, for each coalition
		$C \in \pow{L}$ the coalition $\AgnSet[ - C]$ is using a
		punishment strategy against $C$.
		Now, as for Nash equilibria, we need to quantify over $W \subseteq \AgnSet$
		to obtain an automata characterisation: 
		
		$$
		\AName[\sNE](\GName) = \bigoplus_{W \subseteq \AgnSet} \AName_{W}
		\otimes \AName[\overline{\pow{L}}].
		$$
		
		The proof of correctness of this construction and its complexity is 
		as for Theorem~\ref{thm:ldlNEnonempt}.
		Moreover, the same result can be obtained also with \QPLDLF objectives, as well as \RMGF games.
		Also, observe that the reduction from \LDLF synthesis provided for
		\textsc{NE Non-Emptiness} with \QPLDLF goals can be reused with the same 
		construction for the case of \textsc{sNE Non-Emptiness}. 
		Formally, we have the following result. 
		
		\begin{theorem}
			\label{thm:stronNE}
			\textsc{sNE Non-Emptiness} is in 2EXPTIME for both \iBGF and \RMGF, by using either \LDLF or \QPLDLF goals.
			In particular, for games with \QPLDLF goals, the problem is 2EXPTIME-complete. 
			In addition, the set of strong Nash equilibria for games with either kind of goals is $\omega$-regular. 
		\end{theorem}
	
		\begin{remark}
			The reader might notice that the automata product in the definition of $\AName[sNE](\GName)$ contains a number of factors that is exponential in the number of agents.
			However, this blow-up does not affect the complexity of its emptiness problem.
			Indeed, every single factor is already of size double exponential.
			Therefore, a multiplication of an exponential number of doubly exponential sized automata is still of size doubly exponential.
		\end{remark}

	\end{paragraph}

\section{Concluding Remarks}
	\label{secn:final}
	
	\paragraph*{Logic-based Multi-player Games Revisited}
	In the introduction section it was pointed out that the RMG/iBG and \iBGF frameworks rely on different automata techniques, and that \RMGF/\iBGF is better suited in certain scenarios. However, it is not the case that the \RMGF/\iBGF frameworks generalise iBGs. 
	Indeed, it should be noted that they are incomparable models. For instance, while iBG cannot be used to reason about games with goals over finite traces, \RMGF/\iBGF cannot be used to reason about games with goals over infinite traces {\em only}, that is, regardless of the satisfaction of players' goals in the associated finite traces; from a logical point of view, while RMG/iBG considers \LTL, \RMGF/\iBGF can handle goals in \LDLF, which on finite traces is strictly more expressive than \LTLF~\cite{DV13}, and also than LTL over finite traces. 
	However, as shown here using new automata techniques, the complexities of some problems in each game model coincide in the worst case for many variants of these different kinds of games.
	
	\paragraph*{Automata for Linear Dynamic Logic}
	Logic, games, and automata are intimately related; see, {\em e.g.}, \cite{BozzelliMP15,FijalkowPS13} and references therein for examples. 
	We build upon the automata constructions for Linear Dynamic Logic (\LDLF
	\cite{DV13,DV15}), which were introduced to solve the satisfiability and
	synthesis problems for \LDLF over finite traces. Specifically, 
	we have initially used such
	constructions to translate \LDLF formulae to alternating automata on finite
	words (AFW) and, based on them, we have defined {\em new and optimal automata
		constructions} that characterise the existence of (strong) Nash equilibria on
	top of the standard Boolean games framework.
	
	However, even though most, but not all, of the automata
	constructions we have presented in this paper are optimal, 
	they still enjoy two useful properties.
	Firstly, that they are strongly based on automata on finite words, with only
	an extension to deal with infinite runs, a feature that could be used
	a lot further.
	Secondly, that such automata constructions recognise the sets of Nash runs,
	which, as shown in this paper, makes them extremely useful from a semantic point of
	view. Indeed, using other automata approaches, {\em e.g.}, for iBGs, our 
	expressiveness results do not easily follow. 
		
	In addition, the automata constructions defined in this paper 
	can be modified to reason about other game settings, making it a rather
	widely adaptable reasoning technique. 
	For instance, we believe it is also possible to extend {\em some} of the results we have obtained 
	to two-player games with imperfect information. This should be possible, 
	in some cases, using recent automata 
	constructions to reason about \LDLF formulae under partial observation~\cite{DV16}.
		
	\paragraph*{Imperfect Information}
	Following the research line delineated in~\cite{DV13,DV15,DV16}, one might wonder about the complexity of solving \RMGF/\iBGF in the context of imperfect information.
	It is important to notice that the synthesis problems for \LTL for both perfect and imperfect information is decidable~\cite{PR89,KV00}.
	On the other hand, the \textsc{NE-Nonemptiness} problems for games having \LTL goals is decidable for the perfect information case~\cite{GHW15b}, but undecidable for the imperfect information case~\cite{GPW16}.
	Similarly, regarding \LDLF goals, the synthesis problem is decidable for both perfect~\cite{DV15} and imperfect~\cite{DV16} information, while we here prove that the \textsc{NE-Nonemptiness} problem is decidable.
	This suggests that the same problem might be undecidable under the imperfect information assumption, as it is for conventional iBG{s}.
	However, as the expressive power of \LDLF is incomparable with the one of \LTL, it is not clear whether the undecidability proof (which strongly relies on the expressiveness of \LTL) can be retained in this case.
	Moreover, it has been shown that for specific cases of imperfect information in games with \LTL objectives, the problem might be decidable~\cite{BMM17,BMMRV17}.
	For this reason, we plan to address this question in future work.
	
	\paragraph*{Verification and Equilibria in Logical Form}
	Here we have addressed the ``rational verification'' problem~\cite{rv-aaai16} of multi-agent systems using an automata-theoretic approach.
	We believe that these automata constructions may be used as the underlying
	models of a variant of strategy logic~\cite{MMPV14} (\SL) based on \LDLF over finite traces.
	We will do so in future work.
	\SL is not the only option to investigate. Many other logics for strategic reasoning can be found in the literature; see, {\em e.g.}, \cite{Goranko14}. 
	We believe this is a promising research direction, which would allow us to reason about the behaviour of games over finite traces using other solution concepts in a uniform logical framework.
	A step toward this direction has been recently taken by the formal methods community~\cite{BLMR18a,BLMR18b}.
	
	\paragraph*{Expressiveness and Automata Characterisations}
	Nash equilibrium, and other solution concepts, have been characterised in other works using a range of automata and logical approaches.
	See, for instance, \cite{ChatterjeeHP10,FismanKL10,MMPV14,KPV15} for some examples.
	Although delivering optimal constructions, it is not clear that using such techniques one can characterise the sets of Nash runs in a game. %
	A key feature and contribution of our automata constructions is that they can be used to characterise such sets of equilibrium runs in a uniform way.

\subsection*{Acknowledgments}

	This paper is an extended version of~\cite{GPW17}.
	We acknowledge with gratitude the financial support of the ERC Advanced Investigator grant 291528 (``RACE'') at Oxford.
	Giuseppe Perelli also acknowledges the financial support of the ERC Consolidator Investigator grant 772459 (``d-SynMA'') at G\"{o}teborg.

	\section*{References}
	
	\bibliography{biblio}

\end{document}